\documentclass[final,onefignum,onetabnum]{siamart171218}
\newcommand{\oomit}[1]{}
\oomit{
\usepackage{amssymb}
\usepackage{caption}
\setcounter{tocdepth}{3}
\usepackage{graphicx}
\usepackage{epsfig}
\usepackage{multirow}
\usepackage{multicol,lipsum}
\usepackage{amsmath}
\usepackage{amsfonts}
\usepackage{amssymb}
\usepackage{hyperref}
\usepackage{bm}
\usepackage[usenames]{color}
\newtheorem{example}{Example}
\newtheorem{definition}{Definition}
\newtheorem{proposition}{Proposition}
\newtheorem{lemma}{Lemma}
\newtheorem{theorem}{Theorem}
\newtheorem{corollary}{Corollary}
\usepackage{algorithm,algorithmic}

   \newtheorem{assumption}{Assumption}
\newtheorem{remark}{Remark}
\newcommand{\oomit}[1]{}

}
\usepackage{braket,amsfonts}

\usepackage{array}
\usepackage{bm}
\usepackage{algorithm,algorithmic}
\usepackage[caption=false]{subfig}
\usepackage{pgfplots}

\newsiamthm{claim}{Claim}
\newsiamremark{remark}{Remark}
\newsiamthm{example}{Example}
\newsiamthm{assumption}{Assumption}
\newsiamremark{hypothesis}{Hypothesis}
\crefname{hypothesis}{Hypothesis}{Hypotheses}

\usepackage{algorithm,algorithmic}

\usepackage{graphicx,epstopdf}

\Crefname{ALC@unique}{Line}{Lines}

\usepackage{amsopn}

\usepackage{etoolbox}
\makeatletter
\patchcmd{\maketitle}
 {\def\@makefnmark}
 {\def\@makefnmark{}\def\useless@macro}
 {}{}
\makeatother

\usepackage{xspace}
\usepackage{multirow}
\usepackage{multicol,lipsum}
\usepackage{bold-extra}
\usepackage[most]{tcolorbox}

\colorlet{texcscolor}{blue!50!black}
\colorlet{texemcolor}{red!70!black}
\colorlet{texpreamble}{red!70!black}
\colorlet{codebackground}{black!25!white!25}

\lstdefinestyle{siamlatex}{%
  style=tcblatex,
  texcsstyle=*\color{texcscolor},
  texcsstyle=[2]\color{texemcolor},
  keywordstyle=[2]\color{texemcolor},
  moretexcs={cref,Cref,maketitle,mathcal,text,headers,email,url},
}

\tcbset{%
  colframe=black!75!white!75,
  coltitle=white,
  colback=codebackground, 
  colbacklower=white, 
  fonttitle=\bfseries,
  arc=0pt,outer arc=0pt,
  top=1pt,bottom=1pt,left=1mm,right=1mm,middle=1mm,boxsep=1mm,
  leftrule=0.3mm,rightrule=0.3mm,toprule=0.3mm,bottomrule=0.3mm,
  listing options={style=siamlatex}
}

\newtcbinputlisting[use counter=example]{\examplefile}[3][]{%
  title={Example~\thetcbcounter: #2},listing file={#3},#1}

\DeclareTotalTCBox{\code}{ v O{} }
{ 
  fontupper=\ttfamily\color{black},
  nobeforeafter,
  tcbox raise base,
  colback=codebackground,colframe=white,
  top=0pt,bottom=0pt,left=0mm,right=0mm,
  leftrule=0pt,rightrule=0pt,toprule=0mm,bottomrule=0mm,
  boxsep=0.5mm,
  #2}{#1}

\patchcmd\newpage{\vfil}{}{}{}
\makeatletter
\def\thanks#1{\protected@xdef\@thanks{\@thanks
        \protect\footnotetext{#1}}}
\makeatother
\title{Synthesizing Robust Domains of Attraction for State-Constrained Perturbed Polynomial Systems}

\author{Bai Xue$^{1,2}$, Qiuye Wang$^{1,2}$, Naijun Zhan$^{1,2}$, Shijie Wang$^3$ and Zhikun She$^3$\thanks{1. State Key Lab. of Computer Science, Institute of Software Chinese Academy of Sciences, Beijing, China (\{xuebai,wangqye,znj\}@ios.ac.cn).\newline 
2. University of Chinese Academy of Sciences, Beijing, China \newline
3. School of Mathematics and Systems Science, Beihang University, Beijing, China (\{by1609109,zhikun.she\}@buaa.edu.cn)}}

\begin{document}
\maketitle

\begin{abstract}
In this paper we propose a novel semi-definite programming based method to compute robust domains of attraction for state-constrained perturbed polynomial systems. A robust domain of attraction is a set of states such that every trajectory starting from it will approach an equilibrium while never violating a specified state constraint, regardless of the actual perturbation. The semi-definite program is constructed by relaxing a generalized Zubov's equation. The existence of solutions to the constructed semi-definite program is guaranteed and there exists a sequence of solutions such that their strict one sub-level sets inner-approximate the interior of the maximal robust domain of attraction in measure under appropriate assumptions. Some illustrative examples demonstrate the performance of our method.
\end{abstract}

\section{Introduction}
\label{sec:intro}
A robust domain of attraction of interest in this paper is a set of states from which the system will finally approach an equilibrium while never breaching a specified state constraint regardless of the actual perturbation. Computing it is a fundamental task in the analysis of dynamical systems such as power systems \cite{abu1981} and turbulence phenomena in fluid dynamics \cite{baggett1997}.

Existing approaches to approximating robust domains of attraction can be classified into non-Lyapunov and Lyapunov based categories. Non-Lyapunov based approaches include, but are not limited to, trajectory reversing methods \cite{genesio1985}, polynomial level-set methods \cite{wang2013} and reachable set computation based methods \cite{Korda13,henrion2014,el2017}. Contrasting with non-Lyapunov based methods, Lyapunov based methods are still dominant in estimating robust domains of attraction, e.g., \cite{vannelli1985,khalil1996,chesi2004estimating,Stefan2010,giesl2015}. Such methods are based on the search of a Lyapunov function $V(\bm{x})$ and  a positive scalar $b$ such that the first-order Lie derivative of $V(\bm{x})$ is negative over the sub-level set $C=\{\bm{x}\mid V(\bm{x})\leq b\}$. Given such $V(\bm{x})$ and $b$, it can be shown that the connected component of $C$ containing the equilibrium is a robust domain of attraction. Generally, the search for Lyapunov functions is non-trivial for nonlinear systems due to the non-constructive nature of the Lyapunov theory, apart from some cases where the Jacobian matrix of the linearized system associated with the nonlinear system of interest is Hurwitz. However, with the advance of real algebraic geometry \cite{parrilo2000,bochnak2013} and polynomial optimization  \cite{putinar93,chesi1999,henrion2005} in the last decades, especially the sum-of-squares (SOS) decomposition technique, finding a Lyapunov function which is decreasing over a given state constraint set can be reduced to a convex programming problem for polynomial systems. This results in a large amount of findings which adopt convex optimization based approaches to the search for polynomial Lyapunov functions, e.g., \cite{papachristodoulou2002,chesi2013,anderson2015} and the references therein. However, if we return to the problem of estimating robust domains of attraction, it resorts to addressing a bilinear semi-definite program, e.g., \cite{jarvis2003,tan2008,topcu2010,giesl2015}, which falls within the non-convex programming framework and is notoriously hard to solve \cite{boyd1997}. Moreover, the existence of solutions to (bilinear) semi-definite programs is not explored theoretically in the existing literature.

Another way to compute Lyapunov functions and estimate robust domains of attraction is based on solving the Zubov's equation \cite{zubov1964}, which is a Hamilton-Jacobi type partial differential equation. Zubov's equation was originally inferred to describe the maximal domain of attraction for nonlinear systems free of perturbation inputs and state constraints. Recently, it was extended to perturbed nonlinear systems in \cite{camilli2001} and further to state-constrained perturbed nonlinear systems in \cite{grune2015}. The appealing aspect of Zubov's method in \cite{grune2015} is that it touches upon the problem of computing the \emph{maximal} robust domain of attraction, whose interior is described exactly via the strict one sub-level set of the unique viscosity solution to a generalized Zubov's equation. However, it is well-known that it is notoriously hard to solve the Zubov type equation generally \cite{zubov1964,grune2015}. \oomit{Existing numerical methods for solving Hamilton-Jacobi type partial differential equations such as Zubov's equation require gridding state and perturbation spaces \cite{camilli2001,mitchell2007,bokanowski2013roc,grune2015}, thereby exhibiting exponential growth in computational complexity
with the number of state and perturbation variables and rendering Zubov's method suitable for systems with dimension less than six \cite{bansal2017,xue2019robust}.}

In this paper we propose a novel semi-definite programming based method to compute robust domains of attraction for state-constrained perturbed polynomial systems. We first customize the generalized Zubov's equation in \cite{grune2015} to characterize the maximal robust domain of attraction based on Kirszbraun's extension theorem for Lipschitz maps. Then we relax the customized Zubov's equation into a system of inequalities and further encode these inequalities in the form of sum-of-squares constraints such that a robust domain of attraction can be generated via solving a semi-definite program. The semi-definite program falls within the convex programming framework and can be solved efficiently in polynomial time via interior-point methods, consequently providing a practical method for computing robust domains of attraction. Under appropriate assumptions, the existence of solutions to the constructed semi-definite program is guaranteed and there exists a sequence of solutions such that their strict one sub-level sets inner-approximate the interior of the maximal robust domain of attraction in measure. Some illustrative examples demonstrate the performance of our method. 

The main contributions of this paper are summarized as follows.
\begin{enumerate}
\item A novel semi-definite programming based method is proposed to synthesize robust domains of attraction. The semi-definite program is constructed based on a customized Zubov's equation of the form in \cite{grune2015}. Unlike \cite{grune2015} which reduces the problem of computing robust domains of attraction to a problem of solving a Zubov type equation, our approach reduces the robust domains of attraction generation problem to a semi-definite programming problem. 
\item Under appropriate assumptions, the existence of solutions to the constructed semi-definite program is guaranteed and there exists a sequence of solutions such that their strict one sub-level sets inner-approximate the interior of the maximal robust domain of attraction in measure.
\end{enumerate}

This paper is structured as follows. In Section \ref{Pre} we introduce basic notations used throughout this paper and the robust domains of attraction generation problem of interest. In Section \ref{roa} we detail our semi-definite programming based method for computing robust domains of attraction for state-constrained perturbed polynomial systems. After evaluating this approach on five illustrative examples in Section \ref{ex}, we conclude our paper  in Section \ref{con}.

\section{Preliminaries}
\label{Pre}
In this section we first formulate the problem of generating robust domains of attraction in Subsection \ref{mrda}, and then introduce the concept of the maximal robust domain of uniform attraction in Subsection \ref{umrda}. The maximal robust domain of uniform attraction, which is the interior of the maximal robust domain of attraction, plays an important role for synthesizing robust domains of attraction.

The following basic notations will be used throughout the rest of this paper: $\mathbb{N}$ denotes the set of non-negative integers.  $\mathbb{R}^n$ denotes the set of $n$-dimensional real vectors. $\mathbb{R}[\cdot]$ denotes the ring of polynomials in variables given by the argument. $\mathbb{R}_k[\cdot]$ denotes the set of real polynomials of degree at most $k$ in variables given by the argument, $k\in\mathbb{N}$.
$C^1(\mathbb{R}^n)$ denotes the set of continuously differentiable functions over $\mathbb{R}^n$. $\|\bm{x}\|$ denotes the 2-norm, i.e., $\|\bm{x}\|=\sqrt{\sum_{i=1}^n x_i^2}$, where $\bm{x}=(x_1,\ldots,x_n)^{\top}$. Vectors are denoted by boldface letters. $\Delta^{\circ}$, $\Delta^c$ and $\overline{\Delta}$ denote the interior, the complement and the closure of a set $\Delta$, respectively. $B(\bm{x},r)$ denotes the closed ball around $\bm{x}$ with radius $r>0$ in $\mathbb{R}^n$.

\subsection{Robust Domains of Attraction}
\label{mrda}
A state-constrained perturbed dynamical system of interest in this paper is of the following form:
\begin{equation}
\label{systems}
\dot{\bm{x}}(t)=\bm{f}(\bm{x}(t),\bm{d}(t)),
\end{equation}
where $\bm{x}(\cdot):[0,\infty)\rightarrow \mathcal{X}$, $\bm{d}(\cdot):[0,\infty)\rightarrow D$, $\mathcal{X}\subset \mathbb{R}^n$ is a bounded open set, $D=\{\bm{d} \in \mathbb{R}^m \mid \bigwedge_{i=1}^{m_D} [h_i^D(\bm{d})-1\leq 0]\}$ is a compact set in $\mathbb{R}^m$ with $h_i^D(\bm{d})\in \mathbb{R}[\bm{d}]$, and $\bm{f}(\bm{x},\bm{d})\in \mathbb{R}[\bm{x},\bm{d}]$, thus satisfying the local Lipschitz condition. 

Denote the set of admissible perturbation inputs as
\begin{equation*}
\mathcal{D}=\{\bm{d}(\cdot)\mid \bm{d}(\cdot): [0,+\infty)\rightarrow D~is~measurable\}.
\end{equation*}
As a consequence, for $\bm{x}_0\in \mathcal{X}$ and $\bm{d}(\cdot)\in \mathcal{D}$, there exists a unique absolutely continuous trajectory $\bm{\phi}_{\bm{x}_0}^{\bm{d}}(t)$ satisfying \eqref{systems} a.e. with $\bm{\phi}_{\bm{x}_0}^{\bm{d}}(0)=\bm{x}_0$ for some time interval $[0,T]$ with $T>0$ \cite{thompson2013}.

Additionally, we have Assumption \ref{property} for system \eqref{systems} throughout this paper.
\begin{assumption}
\label{property}
\begin{enumerate}
\item[(1)] $\bm{f}(\bm{0},\bm{d})=\bm{0}, \forall \bm{d}\in D$, i.e., the fixed point $\bm{x}=\bm{0}$ is invariant under all perturbations. 
\item[(2)] there exist positive constants $C$, $\sigma$, $\overline{r}$ such that
\begin{equation}
\label{epsilon}
\|\bm{\phi}_{\bm{x}_0}^{\bm{d}}(t)\|\leq C e^{-\sigma t}\|\bm{x}_0\|
\end{equation}
for $\bm{x}_0\in B(\bm{0},\overline{r})$ and $\bm{d}(\cdot)\in \mathcal{D}$, i.e., the equilibrium state $\bm{0}$ is uniformly locally exponentially stable for system \eqref{systems}. 
\item[(3)] The equilibrium $\bm{x}=\bm{0}$ resides in 
the interior of the state constraint set $\mathcal{X}$, i.e., there exists $\overline{r}>0$ such that $B(\bm{0},\overline{r})\subset \mathcal{X}$. We may without loss of generality assume that this $\overline{r}$ is the same as that in (2). In addition, we also assume that $\overline{r}$ is sufficiently small such that every trajectory starting within $B(\bm{0},\overline{r})$  will never leave the set $\mathcal{X}$. 
\item[(4)] The set $\mathcal{X}$ is of the following form
\begin{equation*}
\label{h}
\mathcal{X}=
\left\{\bm{x} \in \mathbb{R}^n\middle|\; \bigwedge_{i=1}^{n_{\mathcal{X}}} [h_{i}(\bm{x})< 1]\right\},
\end{equation*}
where $h_{i}(\bm{x})\in \mathbb{R}[\bm{x}]$ with $h_i(\bm{x})\geq 0$ over $\mathbb{R}^n$ and $h_i(\bm{0})=0$. Note that this implies that $\partial \mathcal{X} \subseteq \cup_{i=1}^{n_{\mathcal{X}}}\{\bm{x}\in \mathcal{X}\mid h_i(\bm{x})=1 \}$.
\end{enumerate}
\end{assumption}

\begin{remark}
The assumption that the fixed point $\bm{x}=\bm{0}$ is invariant under all perturbations is relatively conservative. One of application scenarios of this assumption is that the steady state of the real system is fixed however, an uncertain (or, perturbed) model rather than a deterministic model was established to capture behaviors of the system due to incomplete information, e.g., \cite{camilli2001,topcu2010,chesi2013}.
\end{remark}

Systems with locally exponentially stable equilibria are widely studied in the existing literature, e.g., \cite{khalil1996}. Since it is not enough to know that the system will converge to an equilibrium eventually in many applications, there is a need to estimate how fast the system approaches $\bm{0}$. The concept of exponential stability can be used for this purpose \cite{slotine1991}. 

The goal of this paper is to synthesize robust domains of attraction of the origin for system \eqref{systems}. The maximal robust domain of attraction, which is the set of states such that every possible trajectory starting from it will approach the origin while never leaving the state constraint set $\mathcal{X}$, is formally formulated in Definition \ref{RDA}.
\begin{definition}[(Maximal) Robust Domain of Attraction]
\label{RDA}
Denote \[\mathcal{D}_{ad}(\bm{x}_0)=\{\bm{d}(\cdot)\in \mathcal{D}\mid \bm{\phi}_{\bm{x}_0}^{\bm{d}}(t)\in \mathcal{X}\text{ for } t\in[0,\infty)\}.\] The maximal robust domain of attraction $\mathcal{R}$ is defined as
\begin{eqnarray}
\mathcal{R}:= \{\bm{x}_0\in \mathbb{R}^n\mid \mathcal{D}_{ad}(\bm{x}_0)=\mathcal{D}\text{ and } \nonumber  \lim_{t\rightarrow \infty}\bm{\phi}_{\bm{x}_0}^{\bm{d}}(t)=0\text{ for }\bm{d}(\cdot)\in \mathcal{D}\}.
\end{eqnarray}
A robust domain of attraction $\Omega$ is a subset of the maximal robust domain of attraction $\mathcal{R}$, i.e., $\Omega\subseteq \mathcal{R}$.
\end{definition}

\subsection{Robust Domains of Uniform Attraction}
\label{umrda}
In order to relate robust domains of attraction to a Zubov type equation, a uniform version of the maximal robust domain of attraction is presented in \cite{grune2015}. In this subsection we introduce the maximal robust domain of uniform attraction. 

 To this end, we define the distance between a point $\bm{x}\in \mathbb{R}^n$ and a set $A\subset \mathbb{R}^n$ as
$\mathtt{dist}(\bm{x},A):=\inf_{\bm{y}\in A}\|\bm{x}-\bm{y}\|$. Then, for $\varrho\geq 0$, we define the set of $\varrho-$admissible perturbation inputs as
\begin{equation*}
\mathcal{D}_{ad,\varrho}(\bm{x}_0):=\{\bm{d}(\cdot)\in \mathcal{D}\mid \mathtt{dist}(\bm{\phi}_{\bm{x}_0}^{\bm{d}}(t),\mathcal{X}^c)>\varrho\text{ for } t\in [0,\infty)\}.
\end{equation*}
Note that $\mathcal{D}_{ad,0}(\bm{x}_0)=\mathcal{D}_{ad}(\bm{x}_0)$. The maximal robust domain of uniform attraction is then defined by
\begin{equation*}
\label{uni}
\mathcal{R}_0:=\left\{\bm{x}_0\in \mathbb{R}^n\middle|\;
\begin{aligned}
&\text{there exists }\varrho > 0 \text{ with~} \mathcal{D}_{ad,\varrho}(\bm{x}_0)=\mathcal{D} \text{~and}\\
 & \text{~there exists a function } \beta(t) \text{ satisfying }\\
 & \lim_{t\rightarrow \infty}\beta(t)=0 \text{ with }\|\bm{\phi}_{\bm{x}_0}^{\bm{d}}(t)\|\leq \beta(t) \text{ for }\\
&t\in [0,\infty) \text{ and }\bm{d}(\cdot)\in \mathcal{D}.
\end{aligned}
\right\}.
\end{equation*}
The set $\mathcal{R}_0$ is a uniform version of $\mathcal{R}$ in the sense that for every initial state $\bm{x}_0\in \mathcal{R}_0$ the trajectories have a positive distance of at least $\varrho$ to $\mathcal{X}^c$ and converge  towards $\bm{0}$ with a speed characterized by $\beta(t)$. Neither $\varrho$ or $\beta(t)$ depends on $\bm{d}(\cdot)\in \mathcal{D}$. The relation between $\mathcal{R}_0$ and $\mathcal{R}$ is uncovered in Lemma \ref{relation}.

\begin{lemma}\cite{grune2015}
\label{relation}
The following two statements hold:
\begin{enumerate}
\item $\mathcal{R}_0$ is open.
\item $\mathcal{R}_0=\mathcal{R}^{\circ}$.
\end{enumerate}
\end{lemma}
$\mathcal{R}_0=\mathcal{R}^{\circ}$ implies that $\mathcal{R}_0$ and $\mathcal{R}$ coincide except for a set with void interior.

\section{Computation of Robust Domains of Attraction}
\label{roa}
In this section we detail our method for synthesizing robust domains of attraction. Subsection \ref{re} presents an auxiliary system, to which the global solution over $t\in [0,\infty)$ exists for every $\bm{x}\in \mathbb{R}^n$, and then Subsection \ref{srda} presents a customized Zubov's equation for state-constrained perturbed polynomial systems as well as a semi-definite programming based method for synthesizing robust domains of attraction. Finally, we show in Subsection \ref{AOE} that the constructed semi-definite program is able to generate a convergent sequence of robust domains of attraction to the maximal robust domain of uniform attraction in measure under appropriate assumptions.

\subsection{System Reformulation}
\label{re}
The system reformulation part in this subsection is similar to that in \cite{xue2019inner,xue2019robust}. Different from the present work, the problems of computing inner-approximations of backward reachable sets over finite time horizons and robust invariant sets over the infinite time horizon are respectively considered in \cite{xue2019inner,xue2019robust} based on relaxing Hamilton-Jacobi type partial differential equations. For self-containedness and ease of understanding, we also give it an appropriate description in this section. As $\bm{f}\in \mathbb{R}[\bm{x},\bm{d}]$ in system \eqref{systems}, $\bm{f}$ is only locally Lipschitz continuous over $\bm{x}$. Therefore, the existence of a global solution $\bm{\phi}_{\bm{x}_0}^{\bm{d}}(t)$ over $t\in [0,\infty)$ to system \eqref{systems} is not guaranteed for any initial state $\bm{x}_0\in \mathbb{R}^n$ and any perturbation input $\bm{d}(\cdot)\in \mathcal{D}$ \cite{csikja2007blow}. However, the existence of global solutions is a prerequisite for constructing the Zubov's equation, to which the strict one sub-level set of the viscosity solution characterizes the maximal robust domain of uniform attraction. In this subsection we construct an auxiliary system, to which the global solution over $t\in [0,\infty)$ with any initial state $\bm{x}_0\in \mathbb{R}^n$ and any perturbation input $\bm{d}(\cdot)\in \mathcal{D}$ exists. Also, its solution coincides with the solution to system \eqref{systems} over a compact set
\begin{equation}
\label{B}
B(\bm{0},R)=\{\bm{x}\in \mathbb{R}^n\mid h(\bm{x})\geq 0\},
\end{equation}
 where $h(\bm{x})= R-\|\bm{x}\|^2$. The compact set $B(\bm{0},R)$ is chosen to satisfy $\mathcal{X}\subset B(\bm{0},R)$ and $\partial \mathcal{X}\cap \partial B(\bm{0},R)=\emptyset$.
Such $R$ in \eqref{B} exists since $\mathcal{X}$ is a bounded set in $\mathbb{R}^n$. The set $B(\bm{0},R)$ in \eqref{B} plays three important roles in our semi-definite programming based approach, which will be shown in Subsection \ref{srda}.
\begin{enumerate}
\item The condition $\mathcal{X}\subseteq B(\bm{0},R)$ guarantees that the maximal robust domain of uniform attraction $\mathcal{R}_0$ for system \eqref{systems} can be exactly characterized by trajectories of the auxiliary system \eqref{sys1}, as formulated in Proposition \ref{eqiva}.
\item The condition $\partial \mathcal{X}\cap \partial B(\bm{0},R)=\emptyset$ assures that the strict one sub-level set of the approximating polynomial returned by solving the semi-definite program \eqref{sos} in Subsection \ref{srda} is a robust domain of attraction. It is useful in justifying Theorem \ref{inner} in Subsection \ref{srda}.
\item The condition that $h(\bm{x})$ is of the form $R-\|\bm{x}\|^2$, i.e., $h(\bm{x})=R-\|\bm{x}\|^2$, is used to guarantee the existence of solutions to the semi-definite program \eqref{sos} in Subsection \ref{srda} under appropriate conditions. It is reflected in justifying Theorem \ref{existence} in Subsection \ref{AOE}.
\end{enumerate}

The auxiliary system is of the following form:
 \begin{equation}
\label{sys1}
\dot{\bm{x}}(t)=\bm{F}(\bm{x}(t),\bm{d}(t)),
 \end{equation}
where $\bm{F}(\bm{x},\bm{d}):\mathbb{R}^n\times D\rightarrow \mathbb{R}^n$, which is globally Lipschitz continuous over $\bm{x}\in \mathbb{R}^n$ uniformly over $\bm{d}\in D$ with Lipschitz constant $L_{\bm{f}}$, i.e.,
\begin{equation}
\label{LIP}
\|\bm{F}(\bm{x}_1,\bm{d})-\bm{F}(\bm{x}_2,\bm{d})\|\leq L_{\bm{f}}\|\bm{x}_1-\bm{x}_2\|
\end{equation}
 for $\bm{x}_1,\bm{x}_2\in \mathbb{R}^n$ and $\bm{d}\in D$, where $L_{\bm{f}}$ is the Lipschitz constant of $\bm{f}$ over $B(\bm{0},R)$. Moreover, $\bm{F}(\bm{x},\bm{d})=\bm{f}(\bm{x},\bm{d})$ over $B(\bm{0},R)\times D$, implying that the trajectories governed by system \eqref{sys1} coincide with the ones generated by system \eqref{systems} over the set $B(\bm{0},R)$.

 The existence of system \eqref{sys1} is guaranteed by Kirszbraun's theorem \cite{fremlin2011}, which is recalled as Theorem \ref{kri}.
 \begin{theorem}[Kirszbraun's Theorem]
 \label{kri}
 Let $A\subset \mathbb{R}^k$ be a set and $\bm{f}':A\rightarrow \mathbb{R}^n$ a function, where $k\geq 1$ is an integer. Suppose there exists $\gamma\geq 0$  such that $\|\bm{f}'(\bm{z}_1)-\bm{f}'(\bm{z}_2)\|\leq \gamma \|\bm{z}_1-\bm{z}_2\|$ for $\bm{z}_1,\bm{z}_2\in A$. Then there is a function $\bm{F}':\mathbb{R}^k\rightarrow \mathbb{R}^n$ such that $\bm{F}'(\bm{z})=\bm{f}'(\bm{z})$ for $\bm{z}\in A$ and $\|\bm{F}'(\bm{z}_1)-\bm{F}'(\bm{z}_2)\|\leq \gamma \|\bm{z}_1-\bm{z}_2\|$ for $\bm{z}_1,\bm{z}_2\in \mathbb{R}^k$.
 \end{theorem}

For instance, $\bm{F}(\bm{x},\bm{d})= \inf_{\bm{y}\in B(\bm{0},R)}\big(\bm{f}(\bm{y},\bm{d}) + \bm{z} L_{\bm{f}} \cdot \|\bm{x}-\bm{y}\|\big)$ satisfies \eqref{sys1}, where $\bm{z}$ is an $n$-dimensional vector with each component equal to one.

 Thus, for any pair $(\bm{d}(\cdot),\bm{x}_0)\in \mathcal{D}\times  \mathbb{R}^n$, there exists a unique absolutely continuous trajectory $\bm{x}(t)=\bm{\psi}_{\bm{x}_0}^{\bm{d}}(t)$ satisfying \eqref{sys1} a.e. with $\bm{x}(0)=\bm{x}_0$ for $t\in [0,\infty)$. This requirement is the basis of deriving the Zubov's equation in \cite{grune2015}. Moreover, we have the following proposition stating that the sets $\mathcal{R}$ and $\mathcal{R}_0$ for system \eqref{systems} coincide with the corresponding sets for system \eqref{sys1} as well.
\begin{proposition}
\label{eqiva}
$\mathcal{R}=\{\bm{x}_0\in \mathbb{R}^n\mid \mathcal{D}_{ad}(\bm{x}_0)=\mathcal{D}\text{ and } \lim_{t\rightarrow \infty}\bm{\psi}_{\bm{x}_0}^{\bm{d}}(t)=0\text{ for } \bm{d}(\cdot) \\ \in \mathcal{D}\}$ and
\begin{equation}
\mathcal{R}_0=\left \{\bm{x}_0\in \mathbb{R}^n \middle|\;
\begin{aligned}
&\text{there exists }\varrho > 0 \text{ with~} \mathcal{D}_{ad,\varrho}(\bm{x}_0)=\mathcal{D} \text{~and }\\
&\text{ there exists a function} \lim_{t\rightarrow \infty}\beta(t)=0 \text{ with }\\
&\|\bm{\psi}_{\bm{x}_0}^{\bm{d}}(t)\|\leq \beta(t) \text{~for~} t\in [0,\infty) \text{ and } \bm{d}(\cdot)\in \mathcal{D}.
\end{aligned}
\right\},
\end{equation}
where $\mathcal{R}$ and $\mathcal{R}_0$ are respectively the maximal robust domain of attraction and the maximal robust domain of uniform attraction for system \eqref{systems}.
\end{proposition}
\begin{proof}
Since $\mathcal{X}\subset B(\bm{0},R)$, $\bm{f}(\bm{x},\bm{d})=\bm{F}(\bm{x},\bm{d})$ over $\bm{x}\in \mathcal{X}$ and $\bm{d}\in D$, the trajectories for system \eqref{systems} and \eqref{sys1} coincide in the state constraint set $\mathcal{X}$, it is obvious that the conclusion holds.
\end{proof}

\subsection{Synthesizing Robust Domains of Attraction}
\label{srda}
In this subsection we first follow the procedure in \cite{grune2015} to characterize the maximal robust domain of uniform attraction $\mathcal{R}_0$ of system \eqref{sys1} as the strict one sub-level set of the viscosity solution to a generalized Zubov's partial differential equation, which is specific to state-constrained perturbed polynomial systems. Based on this Zubov type equation, we construct a semi-definite program for generating robust domains of attraction.

For showing the generalized Zubov's equation we first introduce a running cost $g(\bm{x}):\mathbb{R}^n\rightarrow \mathbb{R}$ and a function $h'(\bm{x}): \mathbb{R}^n\rightarrow \mathbb{R}$ satisfying Assumption \ref{assum2} as in \cite{grune2015}.
\begin{assumption}
\label{assum2}
\begin{enumerate}
\item The function $g(\bm{x})$ is a locally Lipschitz continuous function over $\bm{x}\in \mathbb{R}^n$ satisfying
\begin{enumerate}
\item $g(\bm{x})\geq 0$ with $g(\bm{0})=0$, $\forall \bm{x}\in \mathbb{R}^n$;
\item $\inf\{g(\bm{x})\mid \|\bm{x}\|\geq c\}>0$ for every $c>0$;
\item $\int_{0}^{\infty}g(\bm{\psi}_{\bm{x}}^{\bm{d}}(t))dt$ is finite if $t(\bm{x},\bm{d}(\cdot))$ is finite, where $t(\bm{x},\bm{d}(\cdot))=\inf\{t\geq 0\mid \bm{\psi}_{\bm{x}}^{\bm{d}}(t)\in B(\bm{0},\overline{r})\}$.
\end{enumerate}
\item $h'(\bm{x})=-\min_{j\in \{1,\ldots,n_{\mathcal{X}}\}} h'_{j}(\bm{x})$ with  $h'_{j}(\bm{x})=\ln(l[1-h_{j}(\bm{x})])$ and
\begin{equation*}
\begin{split}
    l[1-h_{j}(\bm{x})]=
    \left\{
                \begin{array}{lll}
                  1-h_{j}(\bm{x}),& \text{if }1-h_{j}(\bm{x})> 0 \\
                  0, &\text{otherwise.}\\
                \end{array}
              \right.
              \end{split}
\end{equation*}
with the convention $\ln 0=-\infty$.
\end{enumerate}
\end{assumption}
The function $h'(\bm{x})$ fulfills the requirement in (A4) in \cite{grune2015}, i.e., $h'(\bm{x})$ is locally Lipschitz continuous on $\mathcal{X}$, $h'(\bm{x})=\infty$ iff $\bm{x}\notin \mathcal{X}$, and $\lim_{n\rightarrow \infty}h'(\bm{x}_n)=\infty$ when $\lim_{n\rightarrow \infty}\bm{x}_n=\bm{x}\notin \mathcal{X}$, $h'(\bm{0})=0$. Throughout this paper, the function $g(\bm{x})$ in Assumption \ref{assum2} is chosen as
\begin{equation}
\label{gin}
g(\bm{x})=\left\{
\begin{array}{lll}
&\alpha-\alpha e^{-\frac{\|\bm{x}\|^s}{\texttt{dist}(\bm{x},\mathcal{X}_{\infty}^c)}},\bm{x}\in \mathcal{X}_{\infty}\\
&q(\bm{x}), \bm{x}\in \mathbb{R}^n\setminus\mathcal{X}_{\infty}
\end{array}
\right.,
\end{equation}
where $s$ is a sufficiently large positive scalar value, $\alpha>0$, and $q(\bm{x})\in \mathbb{R}[\bm{x}]$ with $q(\bm{x})>0$ for $\bm{x}\neq \bm{0}$ is a nonnegative polynomial such that
\begin{equation}
\label{initial}
\mathcal{X}_{\infty}=\{\bm{x}\in \mathbb{R}^n \mid q(\bm{x})< \alpha\}
\end{equation}
 is a nonempty robust domain of uniform attraction. The function $q(\bm{x})$ could be a (local) Lyapunov function and there are numerous methods for computing it, e.g., semi-definite programming based methods \cite{papachristodoulou2002,she2013}. In this paper we assume that $q(\bm{x})$ is given. Also, without loss of generality we assume that \[B(\bm{0},\overline{r})\subset \mathcal{X}_{\infty} \text{~and~} \partial B(\bm{0},\overline{r})\cap \partial \mathcal{X}_{\infty}=\emptyset\] since $\overline{r}$ in Assumption \ref{property} can be sufficiently small. The function $g(\bm{x})$ in \eqref{gin} satisfies Assumption \ref{assum2} as well as another important property shown in Lemma \ref{glip}, which will lift the value functions in \eqref{V0} and \eqref{v} shown later to be Lipschitz continuous.

\begin{lemma}
\label{glip}
 The function $g(\bm{x})$ in \eqref{gin} satisfies Assumption \ref{assum2} and
 the following inequality,
\begin{equation}
\label{ss}
|g(\bm{x})-g(\bm{y})|\leq K\max\big\{\|\bm{x}\|^s,\|\bm{y}\|^s,\|\bm{x}\|^{s-2},\|\bm{y}\|^{s-2}\big\} \|\bm{x}-\bm{y}\|, \forall \bm{x}, \bm{y}\in B(\bm{0},\overline{r}),
\end{equation}
where $B(\bm{0},\overline{r})$ is defined in Assumption \ref{property}, $K$ is some positive constant and $s$ is a sufficiently large positive scalar value.
\end{lemma}
\begin{proof}
We first prove that the function $g(\bm{x})$ in \eqref{gin} is locally Lipschitz continuous over $\bm{x}\in \mathbb{R}^n$. It is obvious that for $\bm{x},\bm{y}\in \mathcal{X}_{\infty}^{\circ}$ or $\bm{x},\bm{y}\in \mathbb{R}^n\setminus \mathcal{X}_{\infty}$, there exists a constant $K'$ such that
\[|g(\bm{x})-g(\bm{y})|\leq K'\|\bm{x}-\bm{y}\|.\]
In the following we just need to show that for $\bm{x}\in \partial \mathcal{X}_{\infty}$ (Since $\mathcal{X}_{\infty}$ is open, $\bm{x}\in \mathbb{R}^n\setminus \mathcal{X}_{\infty}$), there exist a neighborhood $B(\bm{x},\sigma')$ of $\bm{x}$ and a constant $K'>0$ such that
\[|g(\bm{x})-g(\bm{y})|\leq K'\|\bm{x}-\bm{y}\|, \forall \bm{y}\in B(\bm{x},\sigma').\]
Since $\bm{x}\in \partial \mathcal{X}_{\infty}$, $g(\bm{x})=q(\bm{x})=\alpha$.  
When $\bm{y}\in (\mathbb{R}^n\setminus \mathcal{X}_{\infty}) \cap B(\bm{x},\sigma')$, there exists a constant $K_1>0$,
\[|g(\bm{x})-g(\bm{y})|=|q(\bm{x})-q(\bm{y})|\leq K_1\|\bm{x}-\bm{y}\|.\]
When $\bm{y}\in \mathcal{X}_{\infty} \cap B(\bm{x}, \sigma')$, we have
\begin{equation}
\label{ii}
\begin{split}
|g(\bm{x})-g(\bm{y})|&=|\alpha-\alpha+\alpha e^{-\frac{\|\bm{y}\|^s}{\texttt{dist}(\bm{y},\mathcal{X}_{\infty}^c)}}|\\
&=|\alpha e^{-\frac{\|\bm{y}\|^s}{\texttt{dist}(\bm{y},\mathcal{X}_{\infty}^c)}}|\leq |\alpha e^{-\frac{\|\bm{y}\|^s}{\|\bm{x}-\bm{y}\|}}|\leq \alpha  \frac{\|\bm{x}-\bm{y}\|}{\|\bm{y}\|^s}.
\end{split}
\end{equation}
The last inequality in \eqref{ii} uses the fact that $e^{-z}\leq \frac{1}{z}$ for $z\geq 0$. Therefore, there exist a neighborhood $B(\bm{x},\sigma')$ of $\bm{x}$ satisfying $\bm{0}\notin B(\bm{x},\sigma')$ (since $\bm{x}\neq \bm{0}$) and a constant $K'>0$,
$|g(\bm{x})-g(\bm{y})|\leq K'\|\bm{x}-\bm{y}\| \text{~holds},$
where $K'\geq \max\{\max_{\bm{y}\in B(\bm{x},\sigma')}\frac{\alpha}{\|\bm{y}\|^s}, K_1\}$. 

Below we show that the function $g(\bm{x})$ satisfies Assumption \ref{assum2}. It is trivial to prove that the function $g(\bm{x})$ satisfies conditions (a) and (b) in Assumption \ref{assum2}. Next we prove that the function $g(\bm{x})$ in \eqref{gin} satisfies (c) in Assumption \ref{assum2}. Suppose that $T=t(\bm{x},\bm{d}(\cdot))<\infty$ and $\bm{y}=\bm{\psi}_{\bm{x}}^{\bm{d}}(T)$. Since $\bm{\psi}_{\bm{x}}^{\bm{d}}(\cdot):\mathbb{R}\rightarrow \mathbb{R}^n$ is continuous over $t$, $g(\bm{\psi}_{\bm{x}}^{\bm{d}}(t))$ is continuous over $t$ as well, implying that $g(\bm{\psi}_{\bm{x}}^{\bm{d}}(t))$ can attain the maximum $M$ over $t\in [0,T]$. Therefore, we have
\begin{equation*}
\begin{split}
&\int_{0}^{\infty}g(\bm{\psi}_{\bm{x}}^{\bm{d}}(t))dt\\
&=\int_{0}^Tg(\bm{\psi}_{\bm{x}}^{\bm{d}}(t))dt+\int_{T}^{\infty}g(\bm{\psi}_{\bm{x}}^{\bm{d}}(t))dt \\
&=\int_{0}^Tg(\bm{\psi}_{\bm{x}}^{\bm{d}}(t))dt+\int_{0}^{\infty}g(\bm{\psi}_{\bm{y}}^{\bm{d}}(t))dt-\int_{0}^{\infty}g(\bm{0})dt\\
&\leq MT+ \int_{0}^{\infty}L_g\|\bm{\psi}_{\bm{y}}^{\bm{d}}(t)\|dt\\
& \leq MT +L_g C \|\bm{y}\| \int_{0}^{\infty} e^{-\sigma t}dt\\
&\leq MT + \frac{L_g C \overline{r}}{\sigma},
\end{split}
\end{equation*}
where $L_g$ is the Lipschitz constant of $g(\bm{x})$ over $\overline{\mathcal{X}}$, and $C,\overline{r},\sigma$ are defined in \eqref{epsilon}. Therefore, $\int_{0}^{\infty}g(\bm{\psi}_{\bm{x}}^{\bm{d}}(t))dt$ is finite if $t(\bm{x},\bm{d}(\cdot))<\infty$.

In the following we show that the function $g(\bm{x})$ satisfies \eqref{ss} over $B(\bm{0},\overline{r})$.

Let $\bm{x},\bm{y}\in B(\bm{0},\overline{r})$, $M_1$ and $M_2$ are two positive constants such that $\texttt{dist}(\bm{z},\mathcal{X}_{\infty}^{c})\leq M_1$ and $M_2\leq \texttt{dist}(\bm{z},\mathcal{X}_{\infty}^{c})$ for $\bm{z}\in B(\bm{0},\overline{r})$ (Such $M_1$ and $M_2$ exist since $\texttt{dist}(\cdot,\mathcal{X}_{\infty}^{c}): B(\bm{0},\overline{r})\rightarrow (0,\infty)$ is Lipschitz continuous and $B(\bm{0},\overline{r})$ is a compact set with $B(\bm{0},\overline{r})\subset \mathcal{X}_{\infty}$ and $\partial B(0,\overline{r})\cap \partial \mathcal{X}_{\infty}=\emptyset$.), and $L_{\texttt{dist}}$ is the Lipschitz constant of the distance function $\texttt{dist}(\cdot,\mathcal{X}_{\infty}^{c})$ over $B(\bm{0},\overline{r})$,
we have that
\begin{equation*}
\begin{split}
&|g(\bm{x})-g(\bm{y})|\\
&=\alpha |e^{-\frac{\|\bm{y}\|^s}{\texttt{dist}(\bm{y},\mathcal{X}_{\infty}^c)}}-e^{-\frac{\|\bm{x}\|^s}{\texttt{dist}(\bm{x},\mathcal{X}_{\infty}^c)}}|\\
&\leq \alpha \big|\frac{\|\bm{y}\|^s}{\texttt{dist}(\bm{y},\mathcal{X}_{\infty}^c)}-\frac{\|\bm{x}\|^s}{\texttt{dist}(\bm{x},\mathcal{X}_{\infty}^c)}\big|\\
&= \frac{\alpha }{\texttt{dist}(\bm{y},\mathcal{X}_{\infty}^c)\texttt{dist}(\bm{x},\mathcal{X}_{\infty}^c)} \big|\|\bm{y}\|^s \texttt{dist}(\bm{x},\mathcal{X}_{\infty}^c)-\|\bm{x}\|^s\texttt{dist}(\bm{y},\mathcal{X}_{\infty}^c)\big|\\
&\leq \frac{\alpha }{M_2^2}\big|\|\bm{y}\|^s \texttt{dist}(\bm{x},\mathcal{X}_{\infty}^c)-\|\bm{x}\|^s\texttt{dist}(\bm{y},\mathcal{X}_{\infty}^c)\big|\\
& \leq \frac{\alpha }{M_2^2}\Big(\big|\|\bm{y}\|^s \texttt{dist}(\bm{x},\mathcal{X}_{\infty}^c)-\|\bm{x}\|^s\texttt{dist}(\bm{x},\mathcal{X}_{\infty}^c) \big | \\
&~~~~~~~~~~~~~~~~~~~~~~~~~~~~~~~~~~~~~~~+\big|\|\bm{x}\|^s\texttt{dist}(\bm{x},\mathcal{X}_{\infty}^c)-\|\bm{x}\|^s\texttt{dist}(\bm{y},\mathcal{X}_{\infty}^c)\big|\Big)\\
&\leq \frac{\alpha }{M_2^2}\Big(M_1\big|\|\bm{x}\|^s-\|\bm{y}\|^s\big|+\|\bm{x}\|^s\big| \texttt{dist}(\bm{y},\mathcal{X}_{\infty}^c)-\texttt{dist}(\bm{x},\mathcal{X}_{\infty}^c)\big|\Big)\\
&\leq \frac{\alpha }{M_2^2}\Big(M_1 M'\max\{\|\bm{x}\|^{s-2},\|\bm{y}\|^{s-2}\}\big|\|\bm{x}\|-\|\bm{y}\|\big|+\|\bm{x}\|^sL_{\texttt{dist}}\big| \|\bm{y}\|-\|\bm{x}\|\big|\Big)\\
&\leq K \max\{\|\bm{x}\|^{s-2},\|\bm{y}\|^{s-2},\|\bm{x}\|^{s},\|\bm{y}\|^{s}\}\big|\|\bm{x}\|-\|\bm{y}\|\big|,
\end{split}
\end{equation*}
where $M'=s\max_{\bm{z}\in B(\bm{0},\overline{r})}\sum_{i=1}^n |z_i|$ and $K=\frac{\alpha \max\{M_1 M',L_{\texttt{dist}}\}}{M_2^2}$. The inequality $\big|\|\bm{x}\|^s-\|\bm{y}\|^s\big|\leq M' \max\{\|\bm{x}\|^{s-2},\|\bm{y}\|^{s-2}\}\big|\|\bm{x}\|-\|\bm{y}\|\big|$ is obtained in the following way: 
\begin{equation*}
\begin{split}
&\big|\|\bm{x}\|^s-\|\bm{y}\|^s\big|\\
&=s\|\bm{\xi}\|^{s-1}\cdot  \big|\|\bm{x}\|-\|\bm{y}\|\big|\\
&=s\|\bm{\xi}\|^{s-2}\cdot  \|\bm{\xi}\| \cdot \big|\|\bm{x}\|-\|\bm{y}\|\big|\\
&\leq M'\max\{\|\bm{x}\|^{s-2},\|\bm{y}\|^{s-2}\}\big|\|\bm{x}\|-\|\bm{y}\|\big|,
\end{split}
\end{equation*}
where $\|\bm{x}\|=\sqrt{\sum_{i=1}^n x_i^2},$ $\bm{\xi}=\lambda \bm{x}+(1-\lambda) \bm{y}$ and $\lambda $ is a constant falling within $(0,1)$. 

The proof is completed. 
\end{proof}

Denote
 \begin{equation}
 \label{V0}
 \begin{split}
 V(\bm{x}):=\sup_{\bm{d}(\cdot)\in \mathcal{D}}\sup_{t\in [0,\infty)}\big\{\int_{0}^t g(\bm{\psi}_{\bm{x}}^{\bm{d}}(\tau))d\tau+h'(\bm{\psi}_{\bm{x}}^{\bm{d}}(t))\big\}
 \end{split}
 \end{equation}
and 
\begin{equation}
\label{v}
 v(\bm{x}):=1-e^{-\delta V(\bm{x})}=\sup_{\bm{d}(\cdot)\in \mathcal{D}}\sup_{t\in [0,\infty)}\big\{1-e^{\delta \tilde{V}}\big\},
\end{equation}
where
$\tilde{V}=-\int_{0}^t g(\bm{\psi}_{\bm{x}}^{\bm{d}}(\tau))d\tau-h'(\bm{\psi}_{\bm{x}}^{\bm{d}}(t))$ and $\delta$ is some positive constant.

According to Theorem 3.1 in \cite{grune2015}, we have the following conclusion that
\begin{enumerate}
\item $\mathcal{R}_0=\{\bm{x}\in \mathbb{R}^n\mid V(\bm{x})<\infty\}=\{\bm{x}\in \mathbb{R}^n\mid v(\bm{x})<1\}$.
\item $V(\bm{x})$ in \eqref{V0} is continuous on $\mathcal{R}_0$. In addition, $\lim_{n\rightarrow \infty}V(\bm{x}_n)=\infty$ if $\lim_{n\rightarrow \infty}\bm{x}_n=\bm{x}\notin \mathcal{R}_0$ or $\lim_{n\rightarrow \infty}\|\bm{x}_n\|=\infty.$
\end{enumerate}

According to Proposition 4.2 in \cite{grune2015}, $V(\bm{x})$ and $v(\bm{x})$ satisfy the following dynamic programming principle.
\begin{lemma}
\label{dymmmm}
Assume that $G(\bm{x},t,\bm{d}(\cdot))=\int_{0}^t g(\bm{\psi}_{\bm{x}}^{\bm{d}}(\tau))d\tau$.
Then the following assertions are satisfied:
\begin{enumerate}
\item for $\bm{x}\in \mathcal{R}_0$ and $t\geq 0$, we have:
      \begin{equation*}
      \label{dynam1}
      \begin{split}
      V(\bm{x})=\sup_{\bm{d}(\cdot)\in \mathcal{D}}\max\Big\{G(\bm{x},t,\bm{d}(\cdot))+&V(\bm{\psi}_{\bm{x}}^{\bm{d}}(t)), \sup_{\tau\in [0,t]}\{G(\bm{x},\tau,\bm{d}(\cdot))+h'(\bm{\psi}_{\bm{x}}^{\bm{d}}(\tau))\}\Big\}.
      \end{split}
      \end{equation*}
\item for $\bm{x}\in \mathbb{R}^n$ and $t\geq 0$, we have:
 \begin{equation*}
      \label{dynam2}
      \begin{split}
      v(\bm{x})=\sup_{\bm{d}(\cdot)\in \mathcal{D}}\max\Big\{1+&(v(\bm{\psi}_{\bm{x}}^{\bm{d}}(t))-1)e^{-\delta G(\bm{x},t,\bm{d}(\cdot))}, \\
      &\sup_{\tau\in [0,t]}\{1-e^{- \delta G(\bm{x},\tau,\bm{d}(\cdot))-\delta h'(\bm{\psi}_{\bm{x}}^{\bm{d}}(\tau))}\}\Big\}.
      \end{split}
      \end{equation*}
\end{enumerate}
\end{lemma}

We further exploit the Lipschitz continuity property of $V(\bm{x})$ and $v(\bm{x})$. The Lipschitz continuity property of $v(\bm{x})$ plays a key role in guaranteeing the existence of solutions to the constructed semi-definite program \eqref{sos} theoretically, which will be introduced later.
\begin{lemma}
\label{conti}
Under Assumption \ref{property} and Assumption \ref{assum2}, then
\begin{enumerate}
\item $V(\bm{x})$ in \eqref{V0} is locally Lipschitz continuous over $\mathcal{R}_0$. 
\item \oomit{if $\delta\geq \frac{L_{\bm{f}}}{\alpha}$ in \eqref{v}, where $\alpha$ is defined in \eqref{initial},} $v(\bm{x})$ in \eqref{v} is locally Lipschitz continuous over $\mathbb{R}^n$.
\end{enumerate}
\end{lemma}
\begin{proof}
1. Since $\sup_{\bm{d}(\cdot)\in \mathcal{D}}\|\bm{\psi}_{\bm{x}}^{\bm{d}}(t)-\bm{\psi}_{\bm{y}}^{\bm{d}}(t)\|\leq e^{L_{\bm{f}}t}\|\bm{x}-\bm{y}\|$ for $t\in [0,\infty)$ and $\bm{x},\bm{y}\in \mathbb{R}^n$, we have that for $\bm{x}_0\in \mathcal{R}_0$ and $t\in [0,\infty)$, there exist $\delta_{\bm{x}_0,t}>0$ and $\rho>0$ such that $B(\bm{x}_0,\delta_{\bm{x}_0,t})\subseteq \mathcal{X}$ and $\mathtt{dist}(\bm{\psi}_{\bm{y}}^{\bm{d}}(\tau), \mathcal{X}^c)\geq \frac{\rho}{2}$ for $\bm{y}\in B(\bm{x}_0,\delta_{\bm{x}_0,t})$, $\bm{d}(\cdot)\in \mathcal{D}$ and $\tau\in [0,t]$.

Therefore, for $\bm{y}\in B(\bm{x}_0,\delta_{\bm{x}_0})$, where $B(\bm{x}_0,\delta_{\bm{x}_0})\subseteq \mathcal{R}_0$, we obtain that
 \begin{equation*}
 \begin{split}
 &|V(\bm{x}_0)-V(\bm{y})|\\
 &\leq \sup_{\bm{d}(\cdot)\in \mathcal{D}}\sup_{t\in [0,\infty)}\Big(\int_{0}^t |g(\bm{\psi}_{\bm{x}_0}^{\bm{d}}(\tau))-g(\bm{\psi}_{\bm{y}}^{\bm{d}}(\tau))|d\tau +|h'(\bm{\psi}_{\bm{x}_0}^{\bm{d}}(t))-h'(\bm{\psi}_{\bm{y}}^{\bm{d}}(t))|\Big)\\
 &\leq \sup_{\bm{d}(\cdot)\in \mathcal{D}}\int_{0}^{\infty} |g(\bm{\psi}_{\bm{x}_0}^{\bm{d}}(\tau))-g(\bm{\psi}_{\bm{y}}^{\bm{d}}(\tau))|d\tau+\sup_{\bm{d}(\cdot)\in \mathcal{D}}\sup_{t\in [0,\infty)}|h'(\bm{\psi}_{\bm{x}_0}^{\bm{d}}(t))-h'(\bm{\psi}_{\bm{y}}^{\bm{d}}(t))|.
 \end{split}
 \end{equation*}

According to Lemma \ref{glip}, \[|g(\bm{x}')-g(\bm{y}')|\leq K\max\{\|\bm{x}'\|^s,\|\bm{y}'\|^s,\|\bm{x}'\|^{s-2},\|\bm{y}'\|^{s-2}\} \|\bm{x}'-\bm{y}'\|,\forall \bm{x}',\bm{y}'\in B(\bm{0},\overline{r}).\] Thus, analogous to the proof of Proposition 4.2 in \cite{camilli2001}, we obtain that
\begin{equation}
\label{lip1}
\begin{split}
&\sup_{\bm{d}(\cdot)\in \mathcal{D}}\int_{0}^{\infty} |g(\bm{\psi}_{\bm{x}_0}^{\bm{d}}(\tau))-g(\bm{\psi}_{\bm{y}}^{\bm{d}}(\tau))|d\tau \leq L_S\|\bm{x}_0-\bm{y}\|,
\end{split}
\end{equation}
 where $L_S$ is some positive constant.

As to $\sup_{\bm{d}(\cdot)\in \mathcal{D}}\sup_{t\in [0,\infty)}|h'(\bm{\psi}_{\bm{x}_0}^{\bm{d}}(t))-h'(\bm{\psi}_{\bm{y}}^{\bm{d}}(t))|$, following the proof of (ii) of Theorem 3.1 in \cite{grune2015} and the fact that $B(\bm{x}_0,\delta_{\bm{x}_0})$ is compact (it indicates that there exists $T'\in [0,\infty)$ such that $\bm{\psi}_{\bm{y}}^{\bm{d}}(t) \in B(\bm{0},\overline{r})$ for $t\geq T'$, $\bm{d}(\cdot)\in \mathcal{D}$ and $\bm{y}\in B(\bm{x}_0,\delta_{\bm{x}_0})$.), we have that there exists a non-negative constant $T$, which is independent of $\bm{d}(\cdot)\in \mathcal{D}$ and $\bm{y}\in B(\bm{x}_0,\delta_{\bm{x}_0})$, such that
$$\sup_{\bm{d}(\cdot)\in \mathcal{D}}|h'(\bm{\psi}_{\bm{x}_0}^{\bm{d}}(T_{\bm{d}}))-h'(\bm{\psi}_{\bm{y}}^{\bm{d}}(T_{\bm{d}}))|=\sup_{\bm{d}(\cdot)\in \mathcal{D}}\sup_{t\in [0,\infty)}|h'(\bm{\psi}_{\bm{x}_0}^{\bm{d}}(t))-h'(\bm{\psi}_{\bm{y}}^{\bm{d}}(t))|,$$
where $T_{\bm{d}}\in [0,T]$. Taking $\delta_0=\min\{\delta_{\bm{x}_0,T},\delta_{\bm{x}_0}\}$, we have that for $\bm{y}\in B(\bm{x}_0,\delta_0)$,
\begin{equation}
\label{lip2}
\sup_{\bm{d}(\cdot)\in \mathcal{D}}\sup_{t\in [0,\infty)}|h'(\bm{\psi}_{\bm{x}_0}^{\bm{d}}(t))-h'(\bm{\psi}_{\bm{y}}^{\bm{d}}(t))|\leq L_{h'}e^{L_{\bm{f}} T}\|\bm{x}_0-\bm{y}\|,
\end{equation}
where $L_{\bm{f}}$ is defined in \eqref{LIP} and  $L_{h'}$ is the Lipschitz constant of the function $h'$ over the compact set $\overline{\Omega(B(\bm{x}_0,\delta_0),[0,T])}$, which is the closure of the set of states visited by system \eqref{sys1} starting from $B(\bm{x}_0,\delta_0)$ within the time interval $[0,T]$ and thus is a subset of the set $\mathcal{X}$.

Combining \eqref{lip1} and \eqref{lip2}, we have \[|V(\bm{x}_0)-V(\bm{y})|\leq L \|\bm{x}_0-\bm{y}\|, \forall \bm{y}\in B(\bm{x}_0,\delta_0),\] where $L=L_{h'}e^{L_{\bm{f}} T}+L_S.$ Thus, $V(\bm{x})$ in \eqref{V0} is Lipschitz continuous on the compact set $B(\bm{x}_0,\delta_0)$ and thus $V(\bm{x})$ is locally Lipschitz continuous on $\mathcal{R}_0$.

2. We will prove that $v(\bm{x})$ is locally Lipschitz continuous over $\mathbb{R}^n$ based on $0\leq v(\bm{x})\leq 1$ for $\bm{x}\in \mathbb{R}^n$ and the fact in Lemma \ref{dymmmm} that 
for $\bm{x}\in \mathbb{R}^n$ and $t\geq 0$, 
 \begin{equation*}
      \begin{split}
      v(\bm{x})=\sup_{\bm{d}(\cdot)\in \mathcal{D}}\max\Big\{1+&(v(\bm{\psi}_{\bm{x}}^{\bm{d}}(t))-1)e^{-\delta G(\bm{x},t,\bm{d}(\cdot))}, \\
      &\sup_{\tau\in [0,t]}\{1-e^{- \delta G(\bm{x},\tau,\bm{d}(\cdot))-\delta h'(\bm{\psi}_{\bm{x}}^{\bm{d}}(\tau))}\}\Big\}.
      \end{split}
      \end{equation*}

For $\bm{x},\bm{y}\in \mathbb{R}^n$, we have that for any $\epsilon>0$, there exists a perturbation input $\bm{d}(\cdot)\in \mathcal{D}$ such that
\begin{equation*}
\begin{split}
&    |v(\bm{x})-v(\bm{y})|\\
&=\big|\sup_{\bm{d}'(\cdot)\in \mathcal{D}}\max\Big\{(v(\bm{\psi}_{\bm{x}}^{\bm{d}'}(T))-1)e^{-\delta G(\bm{x},T,\bm{d}'(\cdot))}, \sup_{\tau\in [0,T]}\{-e^{- \delta G(\bm{x},\tau,\bm{d}'(\cdot))-\delta h'(\bm{\psi}_{\bm{x}}^{\bm{d}'}(\tau))}\}\Big\}\\
&-\sup_{\bm{d}'(\cdot)\in \mathcal{D}}\max\Big\{(v(\bm{\psi}_{\bm{y}}^{\bm{d}'}(T))-1)e^{-\delta G(\bm{y},T,\bm{d}'(\cdot))}, \sup_{\tau\in [0,T]}\{-e^{- \delta G(\bm{y},\tau,\bm{d}'(\cdot))-\delta h'(\bm{\psi}_{\bm{y}}^{\bm{d}'}(\tau))}\}\Big\}\big|\\
&\leq \max\Big\{|(v(\bm{\psi}_{\bm{x}}^{\bm{d}}(T))-1)e^{-\delta G(\bm{x},T,\bm{d}(\cdot))}-(v(\bm{\psi}_{\bm{y}}^{\bm{d}}(T))-1)e^{-\delta G(\bm{y},T,\bm{d}(\cdot))}|, \\
&~~~~~~~~~~~~~~~~~~~~~~~~~\sup_{\tau\in [0,T]}|e^{- \delta G(\bm{x},\tau,\bm{d}(\cdot))}e^{-\delta h'(\bm{\psi}_{\bm{x}}^{\bm{d}}(\tau))}-e^{- \delta G(\bm{y},\tau,\bm{d}(\cdot))}e^{-\delta h'(\bm{\psi}_{\bm{y}}^{\bm{d}}(\tau))}|\Big\}+\epsilon\\
&\leq \max\Big\{l |e^{-\delta G(\bm{x},T,\bm{d}(\cdot))}-e^{-\delta G(\bm{y},T,\bm{d}(\cdot))}|, k\sup_{\tau\in [0,T]}|e^{- \delta G(\bm{x},\tau,\bm{d}(\cdot))}-e^{- \delta G(\bm{y},\tau,\bm{d}(\cdot))}|\Big\}+\epsilon\\
&\leq \delta \max\Big\{l |G(\bm{x},T,\bm{d}(\cdot))-G(\bm{y},T,\bm{d}(\cdot))|, k \sup_{\tau\in [0,T]}|G(\bm{x},\tau,\bm{d}(\cdot))-G(\bm{y},\tau,\bm{d}(\cdot))|\Big\}+\epsilon\\
&\leq \delta \max\Big\{l\int_{0}^T |g(\bm{\psi}_{\bm{x}}^{\bm{d}}(\tau))-g(\bm{\psi}_{\bm{y}}^{\bm{d}}(\tau))|d\tau, k\sup_{t\in [0,T]}\int_{0}^t |g(\bm{\psi}_{\bm{x}}^{\bm{d}}(\tau))-g(\bm{\psi}_{\bm{y}}^{\bm{d}}(\tau))|d\tau\Big\}+\epsilon\\
&\leq \delta \max\Big\{l \int_{0}^T |g(\bm{\psi}_{\bm{x}}^{\bm{d}}(\tau))-g(\bm{\psi}_{\bm{y}}^{\bm{d}}(\tau))|d\tau, k \int_{0}^T |g(\bm{\psi}_{\bm{x}}^{\bm{d}}(\tau))-g(\bm{\psi}_{\bm{y}}^{\bm{d}}(\tau))|d\tau\Big\}+\epsilon\\
&\leq \delta \max\Big\{l \int_{0}^T L_{g}|\bm{\psi}_{\bm{x}}^{\bm{d}}(\tau)-\bm{\psi}_{\bm{y}}^{\bm{d}}(\tau)|d\tau, k \int_{0}^T L_{g}|\bm{\psi}_{\bm{x}}^{\bm{d}}(\tau)-\bm{\psi}_{\bm{y}}^{\bm{d}}(\tau)|d\tau\Big\}+\epsilon\\
&\leq \delta \max\Big\{l\int_{0}^T L_{g}e^{L_{\bm{f}} \tau}\|\bm{x}-\bm{y}\|d\tau, k \int_{0}^T L_{g}e^{L_{\bm{f}} \tau} \|\bm{x}-\bm{y}\|d\tau\Big\}+\epsilon\\
&\leq K\|\bm{x}-\bm{y}\|+\epsilon,
\end{split}
\end{equation*}
where $T>0$, $l>0$, $k>0$, $K>0$, and $L_g$ is the Lipschitz constant of the function $g(\bm{x})$ over the compact set $\Omega(B,[0,T])$, which is the closure of the set of states visited by system \eqref{sys1} starting from the set $B$  within the time interval $[0,T]$, and $B$ is a compact set covering $\bm{x}$ and $\bm{y}$.
      
Thus, $v(\bm{x})$ is locally Lipschitz continuous over $\mathbb{R}^n$.
\end{proof}

\oomit{We will prove that $v(\bm{x})$ with $\delta\geq 1$ in \eqref{v} is locally Lipschitz continuous over 
$\bm{x}\in \mathbb{R}^n$. According to the fact that $v(\bm{x})=1-e^{-\delta V(\bm{x})}$ for $\bm{x}\in \mathcal{R}_0$ as well as $v(\bm{x})\equiv 1$ for $\bm{x}\in \mathbb{R}^n \setminus \mathcal{R}_0$, it is obvious that for $\bm{x}_0\in \mathcal{R}_0$ or $\bm{x}_0\in (\mathbb{R}^n\setminus \mathcal{R}_0)^{\circ}$, there exist positive constants $\delta_0$ and $K$ such that 
\[|v(\bm{x}_0)-v(\bm{y})|\leq K\|\bm{x}-\bm{y}\|, \forall \bm{y}\in B(\bm{x}_0,\delta_0).\]

In the following we just need to prove that for $\bm{x}_0\in \partial \mathbb{R}_0$, there exists a neighborhood $B(\bm{x}_0,\delta_0)$ of $\bm{x}_0$ and a constant $K>0$ such that 
\[|v(\bm{x}_0)-v(\bm{y})|\leq K\|\bm{x}-\bm{y}\|, \forall \bm{y}\in B(\bm{x}_0,\delta_0).\]

Since $\bm{x}_0\in \partial \mathcal{R}_0$, $v(\bm{x}_0)=1$. When $\bm{y}\in (\mathbb{R}^n\setminus \mathcal{R}_0)\cap B(\bm{x}_0,\delta_0)$, there exists a constant $K_1>0$, 
\[|v(\bm{x}_0)-v(\bm{y})|=0\leq K_1\|\bm{x}_0-\bm{y}\|.\]

When $\bm{y}\in \mathcal{R}_0\cap B(\bm{x}_0,\delta_0)$ and there exists $j\in \{1,\ldots,n_{\mathcal{X}}\}$ such that $h_j(\bm{x}_0)=1$, we have that 
\begin{equation*}
\begin{split}
|v(\bm{x}_0)-v(\bm{y})|&=|1-1+e^{-\delta V(\bm{y})}|\leq \|\frac{e^{-\delta h'(\bm{y}) }}{\bm{x}_0-\bm{y}}\|\|\bm{x}_0-\bm{y}\|\\
&\leq \|\frac{(1-h_j(\bm{y}))^{\delta}}{\bm{x}_0-\bm{y}}\|\|\bm{x}_0-\bm{y}\|\leq K_2\|\bm{x}_0-\bm{y}\|.
\end{split}
\end{equation*}

When $\bm{y}\in \mathcal{R}_0\cap B(\bm{x}_0,\delta_0)$ and there does not exist $j\in \{1,\ldots,n_{\mathcal{X}}\}$ such that $h_j(\bm{x}_0)=1$,

Therefore, there exists a neighborhood $B(\bm{x}_0,\delta_0)$ of $\bm{x}_0$ and a constant $K>0$ such that $|v(\bm{x}_0)-v(\bm{y})|\leq K\|\bm{x}-\bm{y}\|, \forall \bm{y}\in B(\bm{x}_0,\delta_0)$, where $K\geq \max\{K_1,K_2\}$.}

\oomit{This proof follows the one of Proposition 4.3 in \cite{camilli2001} with small modifications.
Let $L_0$ denote the Lipschitz constant of the value function $V(\bm{x})$ on $\overline{\mathcal{X}_{\infty}}$, whose existence is guaranteed by the above argument. Also, let $L_g$ and $L_{h'}$ denote the Lipschitz constants of $g$ and $h'$ on $\mathcal{R}_0$ respectively. For $\bm{x}\in \mathcal{R}_0$, define
\[\tau(\bm{x},\bm{d}(\cdot)):=\inf\{t\geq 0\mid \bm{\psi}_{\bm{x}}^{\bm{d}}(t)\in \overline{\mathcal{X}_{\infty}}\},\]
 $$T_{\bm{x}}: =\sup_{\bm{d}(\cdot)\in \mathcal{D}}\tau(\bm{x},\bm{d}(\cdot))$$ and observe that $V(\bm{x})\geq \alpha T_{\bm{x}}$. If $\bm{x},\bm{y}\in \mathcal{R}_0$, then for any $\epsilon>0$, there exists a perturbation input $\bm{d}(\cdot)\in \mathcal{D}$ such that
\begin{equation*}
\begin{split}
&|V(\bm{x})-V(\bm{y})|\\
&\leq \max\big\{\int_{0}^{T}|g(\bm{\psi}_{\bm{x}}^{\bm{d}}(t)-g(\bm{\psi}_{\bm{y}}^{\bm{d}}(t))|dt +|V(\bm{\psi}_{\bm{x}}^{\bm{d}}(T))-V(\bm{\psi}_{\bm{y}}^{\bm{d}}(T))|,\\
&~~~~~~~~~~~~\int_{0}^{T}|g(\bm{\psi}_{\bm{x}}^{\bm{d}}(t)-g(\bm{\psi}_{\bm{y}}^{\bm{d}}(t))|dt +\sup_{\tau\in [0,T]}|h'(\bm{\psi}_{\bm{x}}^{\bm{d}}(\tau))-h'(\bm{\psi}_{\bm{y}}^{\bm{d}}(\tau))|\big\}+\epsilon\\
&\leq \max \big\{\int_{0}^{T} L_g e^{L_{\bm{f}} t}\|\bm{x}-\bm{y}\|dt +L_0\|\bm{x}-\bm{y}\|e^{L_{\bm{f}} T},\\
&~~~~~~~~~~~~~~~~~~~~~~~~~~~~~~~~~~~~\int_{0}^{T} L_g e^{L_{\bm{f}} t}\|\bm{x}-\bm{y}\|dt +L_{h'}\|\bm{x}-\bm{y}\|e^{L_{\bm{f}} T}\big\}+\epsilon\\
&\leq (L'_0+\frac{L_g}{L_{\bm{f}}})e^{\frac{L_{\bm{f}}}{\alpha} V(\bm{x})}\|\bm{x}-\bm{y}\|+\epsilon,
\end{split}
\end{equation*}
where $L'_0=\max\{L_0,L_{h'}\}$ and $T=\min\{T_{\bm{x}},T_{\bm{y}}\}$. Therefore, $V(\bm{x})$ is locally Lipschitz continuous in $\mathcal{R}_0$ with Lipschitz constant $(L'_0+\frac{L_g}{L_{\bm{f}}})e^{\frac{L_{\bm{f}}}{\alpha} V(\bm{x})}$.

Let $\phi(\bm{x}) \in C^1(\mathbb{R}^n)$ be such that $v(\bm{x})-\phi(\bm{x})$ has a local maximum at $\bm{x}_0\in \mathcal{R}_0$, where we may assume that $v(\bm{x}_0)-\phi(\bm{x}_0)=0$ and $\phi(\bm{x})\leq 1$, $\forall \bm{x}\in \mathbb{R}^n$. Then $V(\bm{x})-\psi(\bm{x})$ has a local maximum at $\bm{x}_0$ for
$\psi(\bm{x})=-\frac{\ln (1-\phi(\bm{x}))}{\delta}$. Moreover, $\frac{\partial\phi(\bm{x})}{\partial \bm{x}}\mid_{\bm{x}=\bm{x}_0}=\delta \frac{\partial\psi(\bm{x})}{\partial \bm{x}}\mid_{\bm{x}=\bm{x}_0}e^{-\delta \psi(\bm{x}_0)}$ and $\psi(\bm{x}_0)=V(\bm{x}_0)$.

According to Corollary 3 in \cite{Deville1995}, we have 
 $\|\frac{\partial\psi(\bm{x})}{\partial \bm{x}}\|_{\bm{x}=\bm{x}_0}\leq (L'_0+\frac{L_g}{L_{\bm{f}}})e^{\frac{L_{\bm{f}}}{\alpha} V(\bm{x}_0)}$. Therefore, 
\[\|\frac{\partial\phi(\bm{x})}{\partial \bm{x}}\|_{\bm{x}=\bm{x}_0}\leq \delta \|\frac{\partial\psi(\bm{x})}{\partial \bm{x}}\|_{\bm{x}=\bm{x}_0}e^{-\delta V(\bm{x}_0)}\leq \delta (L'_0+\frac{L_g}{L_{\bm{f}}})e^{(\frac{L_{\bm{f}}}{\alpha}-\delta) V(\bm{x}_0)}.\]
 Hence, since $v(\bm{x})\equiv 1$ in $\mathbb{R}^n\setminus \mathcal{R}_0$, we have that
$\|\frac{\partial\phi(\bm{x})}{\partial \bm{x}}\|_{\bm{x}=\bm{x}_0}\leq \delta (L'_0+\frac{L_g}{L_{\bm{f}}})$ for any $\bm{x}_0\in \mathbb{R}^n$ and any $\phi(\bm{x})\in C^1(\mathbb{R}^n)$ satisfying that $v(\bm{x})-\phi(\bm{x})$ has a local maximum at $\bm{x}_0$. This implies that $v(\bm{x})$ is locally Lipschitz continuous in $\mathbb{R}^n$ according to Corollary 3 in \cite{Deville1995}.}

Consequently, combining Theorem 4.4 in \cite{grune2015} and Lemma \ref{conti} we have: 
\begin{theorem}
\label{hj}
The value function $V(\bm{x})$ in \eqref{V0} is the unique locally Lipschitz continuous viscosity solution to the following equation
\begin{equation}
\label{HJB1}
\begin{aligned}
&\min\big\{ \inf_{\bm{d}\in D}\{-\frac{\partial V(\bm{x})}{\partial \bm{x}}\bm{F}(\bm{x},\bm{d})-g(\bm{x})\},V(\bm{x})-h'(\bm{x})\big\}=0, \forall \bm{x}\in \mathcal{R}_0,\\
& V(\bm{0})=0.
\end{aligned}
\end{equation}
Likewise, the value function $v(\bm{x})$ in \eqref{v} is the unique bounded and locally Lipschitz continuous viscosity solution to the generalized Zubov's equation
\begin{equation}
\label{HJB2}
\begin{split}
&\min\big\{\inf_{\bm{d}\in D}\{-\frac{\partial v(\bm{x})}{\partial \bm{x}}\bm{F}(\bm{x},\bm{d})-\delta g(\bm{x})(1-v(\bm{x}))\},\\
&~~~~~~~~~~~~~~~~~~~~~~~~~~~~~~~~~~~~~~~~~v(\bm{x})+e^{- \delta h'(\bm{x})}-1\big\}=0, \forall \bm{x}\in \mathbb{R}^n,\\
&v(\bm{0})=0.
\end{split}
\end{equation}
\end{theorem}

As a direct consequence of \eqref{HJB2}, we have that if a continuously differentiable function $u(\bm{x}):\mathbb{R}^n\rightarrow \mathbb{R}$ satisfies \eqref{HJB2}, then $u(\bm{x})$ satisfies the constraints: 
\begin{equation}
\label{inequa1}
\left\{
\begin{array}{lll}
&-\frac{\partial u(\bm{x})}{\partial \bm{x}}\bm{F}(\bm{x},\bm{d})-\delta g(\bm{x})(1-u(\bm{x}))\geq 0, \forall \bm{x}\in \mathbb{R}^n, \forall \bm{d}\in D,\\
&u(\bm{x})+\min_{j\in \{1,\ldots,n_{\mathcal{X}}\}}(1-h_j(\bm{x}))^{\delta}-1\geq 0, \forall \bm{x}\in \mathcal{X},\\
&u(\bm{x})\geq 1, \forall \bm{x}\in \mathbb{R}^n\setminus \mathcal{X}.
\end{array}
\right.
\end{equation}

\begin{corollary}
\label{relax}
Assume a continuously differentiable function $u(\bm{x}): \mathbb{R}^n \rightarrow \mathbb{R}$ is a solution to \eqref{inequa1}, then $v(\bm{x})\leq u(\bm{x})$ over $\bm{x}\in \mathbb{R}^n$ and consequently $\Omega=\{\bm{x}\mid u(\bm{x})<1\}\subset \mathcal{R}_0$ is a robust domain of uniform attraction.
\end{corollary}
\begin{proof}
It is obvious that \eqref{inequa1} is equivalent to
\begin{equation}
\label{inequa}
\left\{
\begin{array}{lll}
&-\frac{\partial u(\bm{x})}{\partial \bm{x}}\bm{F}(\bm{x},\bm{d})-\delta g(\bm{x})(1-u(\bm{x}))\geq 0, \forall \bm{x}\in \mathbb{R}^n, \forall \bm{d}\in D,\\
&u(\bm{x})+e^{-\delta h'(\bm{x})}-1\geq 0, \forall \bm{x}\in \mathbb{R}^n.
\end{array}
\right.
\end{equation}

If $u(\bm{x})$ is a viscosity super-solution to \eqref{HJB2}, according to the comparison principle in Proposition 4.7 in \cite{grune2015}, $v(\bm{x})\leq u(\bm{x})$ holds. Consequently, $\Omega=\{\bm{x}\mid u(\bm{x})<1\}\subset \mathcal{R}_0$ is a robust domain of uniform attraction. In the following we show that $u(\bm{x})$ is a viscosity super-solution to \eqref{HJB2}

Let's first recall the concept of viscosity super-solution to \eqref{HJB2}. A lower semicontinuous function $u_l(\cdot): \mathbb{R}^n\rightarrow \mathbb{R}$ is a viscosity super-solution of \eqref{HJB2} \cite{grune2015} if for all $\phi(\bm{x})\in C^1(\mathbb{R}^n)$ such that $u_l(\bm{x})-\phi(\bm{x})$ has a local minimum at $\bm{x}_0$, we have
\begin{equation*}
\begin{split}
\min\big\{\inf_{\bm{d}\in D}\{-\frac{\partial \phi(\bm{x})}{\partial \bm{x}}\mid_{\bm{x}=\bm{x}_0}\bm{F}(\bm{x}_0,\bm{d})-&\delta g(\bm{x}_0)(1-u_l(\bm{x}_0))\},\\
&u_l(\bm{x}_0)+e^{- \delta h'(\bm{x}_0)}-1\big\}\geq 0.
\end{split}
\end{equation*}

Since $u(\bm{x})$ satisfies \eqref{inequa}, we just show that $\inf_{\bm{d}\in D}\{-\frac{\partial \phi(\bm{x})}{\partial \bm{x}}\mid_{\bm{x}=\bm{x}_0}\bm{F}(\bm{x}_0,\bm{d})-\delta g(\bm{x}_0)(1-u(\bm{x}_0))\}\geq 0$, where $\phi(\bm{x})\in C^1(\mathbb{R}^n)$ and $u(\bm{x})-\phi(\bm{x})$ has a local minimum at $\bm{x}_0$. Without loss of generality, we assume that $u(\bm{x}_0)-\phi(\bm{x}_0)=0$. There exists $\delta_0>0$ such that
$$u(\bm{x})-\phi(\bm{x})\geq 0, \forall \bm{x}\in B(\bm{x}_0,\delta_0).$$

Suppose that $\inf_{\bm{d}\in D}\{-\frac{\partial \phi(\bm{x})}{\partial \bm{x}}\mid_{\bm{x}=\bm{x}_0}\bm{F}(\bm{x}_0,\bm{d})-\delta g(\bm{x}_0)(1-u(\bm{x}_0))\}\geq 0$ does not hold, i.e., $\inf_{\bm{d}\in D}\{-\frac{\partial \phi(\bm{x})}{\partial \bm{x}}\mid_{\bm{x}=\bm{x}_0}\bm{F}(\bm{x}_0,\bm{d})-\delta g(\bm{x}_0)(1-\phi(\bm{x}_0))\}< 0$. Then there exists $\epsilon>0$ such that
$$\inf_{\bm{d}\in D}\{-\frac{\partial \phi(\bm{x})}{\partial \bm{x}}\mid_{\bm{x}=\bm{x}_0}\bm{F}(\bm{x}_0,\bm{d})-\delta g(\bm{x}_0)(1-\phi(\bm{x}_0))\}=-\epsilon.$$
Further, there exists $\bm{d}_1\in D$ such that
 $$-\frac{\partial \phi(\bm{x})}{\partial \bm{x}}\mid_{\bm{x}=\bm{x}_0}\bm{F}(\bm{x}_0,\bm{d}_1)-\delta g(\bm{x}_0)(1-\phi(\bm{x}_0))\leq -\frac{\epsilon}{2}$$
 and consequently there exists $\delta'>0$ with $\delta'\leq \delta_0$ such that
 $$-\frac{\partial \phi(\bm{x})}{\partial \bm{x}}\bm{F}(\bm{x},\bm{d}_1)-\delta g(\bm{x})(1-\phi(\bm{x}))\leq -\frac{\epsilon}{4}, \forall \bm{x}\in B(\bm{x}_0,\delta').$$
Since $\bm{\psi}_{\bm{x}_0}^{\bm{d}}(t)$ is absolutely continuous over $t$ for $\bm{d}(\cdot)\in \mathcal{D}$, there exists $\theta>0$ such that for $\tau \in [0,\theta]$,
\begin{equation}
\label{con22}
\begin{split}
 -\frac{\partial \phi(\bm{x})}{\partial \bm{x}}\mid_{\bm{x}=\bm{\psi}_{\bm{x}_0}^{\bm{d}'_1}(\tau)}\bm{F}(\bm{\psi}_{\bm{x}_0}^{\bm{d}'_1}(\tau),\bm{d}'_1(\tau))-\delta g(\bm{\psi}_{\bm{x}_0}^{\bm{d}'_1}(\tau))(1-\phi(\bm{\psi}_{\bm{x}_0}^{\bm{d}'_1}(\tau)))\leq -\frac{\epsilon}{4},
 \end{split}
 \end{equation}
where $\bm{d}'_1(\cdot)\in \mathcal{D}$ with $\bm{d}'_1(\tau)=\bm{d}_1$ for $\tau \in [0,\theta]$. Therefore, we have that for $\tau \in [0,\theta]$,
\begin{equation}
\label{con221}
\begin{split}
 -\frac{\partial \phi(\bm{x})}{\partial \bm{x}}\mid_{\bm{x}=\bm{\psi}_{\bm{x}_0}^{\bm{d}'_1}(\tau)}\bm{F}(\bm{\psi}_{\bm{x}_0}^{\bm{d}'_1}(\tau),\bm{d}'_1(\tau))-\delta g(\bm{\psi}_{\bm{x}_0}^{\bm{d}'_1}(\tau))(1-\phi(\bm{\psi}_{\bm{x}_0}^{\bm{d}'_1}(\tau)))<0,
 \end{split}
 \end{equation}
where $\bm{d}'_1(\cdot)\in \mathcal{D}$ with $\bm{d}'_1(\tau)=\bm{d}_1$ for $\tau \in [0,\theta]$.

By applying Gronwall's inequality \cite{gronwall1919} to \eqref{con221} together with \eqref{con22} with the time interval [0, $\theta$], we have that
$$\phi(\bm{x}_0)-1<e^{-\delta G}(\phi(\bm{\psi}_{\bm{x}_0}^{\bm{d}'_1}(\theta))-1),$$
where $G=\int_{0}^{\theta} g(\bm{\psi}_{\bm{x}_0}^{\bm{d}'_1}(t))dt$.
Therefore,
$$u(\bm{x}_0)-1< e^{-\delta  G}(u(\bm{\psi}_{\bm{x}_0}^{\bm{d}'_1}(\theta))-1),$$
which contradicts the fact that
$$-\frac{\partial u(\bm{x})}{\partial \bm{x}}\bm{F}(\bm{x},\bm{d})-\delta g(\bm{x})(1-u(\bm{x}))\geq 0, \forall \bm{x}\in \mathbb{R}^n, \forall \bm{d}\in D.$$
Therefore, we conclude  that $u(\bm{x})$ is a viscosity super-solution to \eqref{HJB2}.
\end{proof}

From Corollary \ref{relax} we observe that a robust domain of attraction can be found by solving \eqref{inequa1} rather than  \eqref{HJB2}. However, $u(\bm{x})$ is required to satisfy \eqref{inequa1} over $\mathbb{R}^n$, which is a strong condition. This requirement renders the search for a continuously differentiable solution to \eqref{inequa1} nonetheless nontrivial. Regarding this issue, we further relax this condition and restrict the search for a continuously differentiable function $u(\bm{x})$ in the compact set $B(\bm{0},R)\setminus \mathcal{X}_{\infty}$, where $B(\bm{0},R)$ is defined in \eqref{B}. Also, since $\bm{F}(\bm{x},\bm{d})=\bm{f}(\bm{x},\bm{d})$ for $(\bm{x},\bm{d})\in B(\bm{0},R)\times D$ and $g(\bm{x})=q(\bm{x})$ for $\bm{x}\in B(\bm{0},R)\setminus \mathcal{X}_{\infty}$, we obtain the following system of constraints:
\begin{equation}
\label{inequa2}
\left\{
\begin{array}{lll}
&-\frac{\partial u(\bm{x})}{\partial \bm{x}}\bm{f}(\bm{x},\bm{d})-\delta q(\bm{x})(1-u(\bm{x}))\geq 0, \forall \bm{x}\in B(\bm{0},R)\setminus \mathcal{X}_{\infty}, \forall \bm{d}\in D,\\
&u(\bm{x})-h_j(\bm{x})\geq 0,j=1,\ldots,n_{\mathcal{X}},\forall \bm{x}\in \overline{\mathcal{X}\setminus \mathcal{X}_{\infty}},\\
&u(\bm{x})- 1\geq 0, \forall \bm{x}\in B(\bm{0},R)\setminus\mathcal{X}.
\end{array}
\right.
\end{equation}

\begin{theorem}
\label{inner}
Let $u(\bm{x})$ be a continuously differentiable solution to \eqref{inequa2} and $\delta$ be a positive value, then $\Omega=\{\bm{x}\in B(\bm{0},R)\mid u(\bm{x})<1\}$ is a robust domain of attraction.
\end{theorem}
\begin{proof}
Firstly, since $\bm{F}(\bm{x},\bm{d})=\bm{f}(\bm{x},\bm{d})$ for $(\bm{x},\bm{d})\in B(\bm{0},R)\times D$ and $g(\bm{x})=q(\bm{x})$ over $\bm{x}\in B(\bm{0},R)\setminus \mathcal{X}_{\infty}$,
\eqref{inequa2} is equivalent to
\begin{equation*}
\label{inequa3}
\left\{
\begin{array}{lll}
&-\frac{\partial u(\bm{x})}{\partial \bm{x}}\bm{F}(\bm{x},\bm{d})-\delta g(\bm{x}) (1-u(\bm{x}))\geq 0, \forall \bm{x}\in \overline{B(\bm{0},R)\setminus \mathcal{X}_{\infty}}, \forall \bm{d}\in D,\\
&u(\bm{x})-h_j(\bm{x})\geq 0, j=1,\ldots,n_{\mathcal{X}},\forall \bm{x}\in \overline{\mathcal{X}\setminus \mathcal{X}_{\infty}},\\
&u(\bm{x})- 1\geq 0, \forall \bm{x}\in B(\bm{0},R)\setminus\mathcal{X}.
\end{array}
\right..
\end{equation*}

Since $u(\bm{x})- 1\geq 0$ for $\bm{x}\in B(\bm{0},R)\setminus \mathcal{X}$, $\Omega\subset \mathcal{X}$ holds. Next we prove that every possible trajectory initialized in the set $\Omega$ will approach the equilibrium state $\bm{0}$ eventually while never leaving the state constraint set $\mathcal{X}$.

Assume that there exist $\bm{y}\in \Omega$, a perturbation input $\bm{d}'(\cdot)\in \mathcal{D}$ and $\tau>0$ such that $$\bm{\psi}_{\bm{y}}^{\bm{d}'}(t) \in \mathcal{X}, \forall t\in [0,\tau)$$ and $$\bm{\psi}_{\bm{y}}^{\bm{d}'}(\tau)\notin \mathcal{X}.$$ Obviously, $\bm{y}\notin \mathcal{X}_{\infty}$ and $\bm{\psi}_{\bm{y}}^{\bm{d}'}(t)\notin \mathcal{X}_{\infty}$ for $t\in [0,\tau]$  since $\mathcal{X}_{\infty}$ is a robust domain of attraction. That is, $$\bm{\psi}_{\bm{y}}^{\bm{d}'}(t)\in B(\bm{0},R)\setminus \mathcal{X}_{\infty}, \forall t\in [0,\tau].$$ Applying Gronwall's inequality \cite{gronwall1919} to $-\frac{\partial u}{\partial \bm{x}}\bm{F}(\bm{x},\bm{d})-\delta g(\bm{x}) (1-u(\bm{x}))\geq 0$ with the time interval [0, $\tau$], we have that
$$u(\bm{y})-1\geq e^{-\delta G}(u(\bm{\psi}_{\bm{y}}^{\bm{d}'}(\tau))-1),$$
where $G=\int_{0}^{\tau}g(\bm{\psi}_{\bm{y}}^{\bm{d}'}(t))dt>0$. Therefore, $u(\bm{\psi}_{\bm{y}}^{\bm{d}'}(\tau))<1$. However, since $\mathcal{X}\subseteq B(\bm{0},R)$ and $\partial \mathcal{X}\cap \partial B(\bm{0},R)=\emptyset$, $\bm{\psi}_{\bm{y}}^{\bm{d}'}(\tau)\in B(\bm{0},R)\setminus \mathcal{X}$ holds and consequently $u(\bm{\psi}_{\bm{y}}^{\bm{d}'}(\tau))\geq 1$. This is a contradiction. Thus, every possible trajectory initialized in the set $\Omega$ never leaves the set $\mathcal{X}$.

Lastly, we prove that every possible trajectory initialized in the set $\Omega$ approaches the equilibrium state $\bm{0}$ eventually. Since every possible trajectory initialized in the set $\mathcal{X}_{\infty}$ approaches the equilibrium state $\bm{0}$ eventually, it is enough to prove that every possible trajectory initialized in the set $\Omega\setminus \mathcal{X}_{\infty}$ will enter the set $\mathcal{X}_{\infty}$ in finite time. Assume that there exist $\bm{y}\in \Omega$ and a perturbation input $\bm{d}'(\cdot)$ such that $\bm{\psi}_{\bm{y}}^{\bm{d}'}(t) \notin \mathcal{X}_{\infty}$ for all $t\geq 0$. According to the second constraint and the third constraint in \eqref{inequa2}, we have $u(\bm{x})\geq 0$ for $\bm{x}\in B(\bm{0},R)\setminus \mathcal{X}_{\infty}$. 
Also, since $\bm{\psi}_{\bm{y}}^{\bm{d}'}(t)\in B(\bm{0},R)$ for all $t\geq 0$, $u(\bm{\psi}_{\bm{y}}^{\bm{d}'}(t))\geq 0$ holds for all $t\geq 0$. Moreover, applying  Gronwall's inequality \cite{gronwall1919} again to $$-\frac{\partial u(\bm{x})}{\partial \bm{x}}\bm{F}(\bm{x},\bm{d})-\delta g(\bm{x}) (1-u(\bm{x}))\geq 0$$ with the time interval [0, $\tau$] for $\tau>0$, we have that
$u(\bm{\psi}_{\bm{y}}^{\bm{d}'}(\tau))<1$. This implies that $$u(\bm{\psi}_{\bm{y}}^{\bm{d}'}(\tau))\in \Omega\setminus \mathcal{X}_{\infty}, \forall \tau \geq 0.$$ Also, since $$\frac{\partial u(\bm{x})}{\partial \bm{x}}\bm{F}(\bm{x},\bm{d})\leq -\delta g(\bm{x})(1-u(\bm{x}))\leq -\alpha \delta (1-u(\bm{x})), \forall \bm{x}\in \Omega\setminus \mathcal{X}_{\infty},$$
we obtain that $$u(\bm{y})-1\geq e^{-\delta \alpha \tau}(u(\bm{\psi}_{\bm{y}}^{\bm{d}'}(\tau))-1).$$ Consequently, we conclude that $$\lim_{\tau\rightarrow \infty}u(\bm{\psi}_{\bm{y}}^{\bm{d}'}(\tau))=-\infty,$$ contradicting the fact that $u(\bm{\psi}_{\bm{y}}^{\bm{d}'}(t))\geq 0$ holds for $t\geq 0$.
Therefore, every possible trajectory initialized in the set $\Omega$ will enter the set $\mathcal{X}_{\infty}$ in finite time and consequently will asymptotically approach the equilibrium state $\bm{0}$.

Therefore, $\Omega $ is a robust domain of attraction.
\end{proof}

When $u(\bm{x})$ in \eqref{inequa2} is a polynomial in $\mathbb{R}[\bm{x}]$, based on the sum-of-squares decomposition for multivariate polynomials, \eqref{inequa2} is recast as the following semi-definite program:
\begin{algorithm}
\begin{algorithmic}
\STATE
\begin{equation}
\label{sos}
\begin{split}
&p^*=\inf \bm{c}\cdot \bm{l}\\
&\texttt{s.t.}\\
&\left\{
\begin{array}{lll}
&-\frac{\partial u(\bm{x})}{\partial \bm{x}}\bm{f}(\bm{x},\bm{d})-\delta q(\bm{x}) (1-u(\bm{x}))=s_0(\bm{x},\bm{d})+s_1(\bm{x},\bm{d})\cdot h(\bm{x})\\
&~~~~~~~~~~~~~~~~~~~~~~~~~+\sum_{i=1}^{m_D}s_{2,i}(\bm{x},\bm{d})\cdot(1-h_i^D(\bm{d}))+s_{3}(\bm{x},\bm{d})\cdot(q(\bm{x})-\alpha),\\
&u(\bm{x})-1=s_{4,j}(\bm{x})+s_{5,j}(\bm{x})\cdot h(\bm{x})+s_{6,j}(\bm{x})\cdot(h_j(\bm{x})-1),\\
&u(\bm{x})-h_j(\bm{x})=s_{7,j}(\bm{x})+s_{8,j}(\bm{x})\cdot h(\bm{x})\\
&~~~~~~~~~~~~~~~~~~~~~~~~~~~~+s_{9,j}(\bm{x})\cdot(q(\bm{x})-\alpha)+\sum_{l=1}^{n_{\mathcal{X}}}s_{10,l,j}(\bm{x})\cdot(1-h_l(\bm{x})),\\
&j=1,\ldots,n_{\mathcal{X}},\\
\end{array}
\right.
\end{split}
\end{equation}
where $\bm{c}\cdot \bm{l}=\int_{B(\bm{0},R)\setminus \mathcal{X}_{\infty}}u(\bm{x})d\mu(\bm{x})$, $\bm{l}$ is the vector of the moments of the Lebesgue measure $\mu(\bm{x})$ over $B(\bm{0},R)\setminus \mathcal{X}_{\infty}$ indexed in the same basis in which the polynomial $u(\bm{x})$ with coefficients $\bm{c}$ is expressed, $B(\bm{0},R)=\{\bm{x}\mid h(\bm{x})\geq 0\}$ and $\mathcal{X}_{\infty}=\{\bm{x}\mid q(\bm{x})< \alpha\}$. $\delta$ is a user-defined positive value. The minimum is over polynomial $u(\bm{x})\in R[\bm{x}]$ and sum-of-squares polynomials  $s_i(\bm{x},\bm{d})$, $i=0,1$, $s_{2,i}(\bm{x},\bm{d})$, $i=1,\ldots,m_D$, $s_{3}(\bm{x},\bm{d})$, $s_{i,j}(\bm{x})$, $s_{10,l,j}(\bm{x})$, $i=4,\ldots,9$, $j,l=1,\ldots,n_{\mathcal{X}}$, of appropriate degree. Since the constraints that polynomials are sum-of-squares can be written explicitly as linear matrix inequalities, and the objective is linear in the coefficients of polynomial $u(\bm{x})$, problem \eqref{sos} is a semi-definite program, which falls within the convex programming framework and can be solved via interior-point methods in polynomial time (e.g., \cite{vandenberghe1996}). 
\end{algorithmic}
\end{algorithm}

According to Theorem \ref{inner}, $\mathcal{R}_{u}=\{\bm{x}\in B(\bm{0},R)\mid u(\bm{x})<1\}$ is a robust domain of attraction, where $u(\bm{x})\in \mathbb{R}[\bm{x}]$ is the solution to \eqref{sos}.

\begin{remark}
$\{\bm{x}\in B(\bm{0},R)\mid u(\bm{x})<1\}$ is still a robust domain of attraction if the origin $\bm{0}$ is asymptotically stable for \eqref{systems} rather than uniformly locally exponentially stable, where $u(\bm{x})$ is the solution to \eqref{sos}. The proof of Theorem \ref{inner} does not require that the equilibrium state $\bm{0}$ is uniformly locally exponentially stable.
\end{remark}

\subsection{Analysis of \eqref{sos}}
\label{AOE}
In this subsection we exploit some properties pertinent to \eqref{sos} and show that there exist solutions to \eqref{sos} under appropriate assumptions. Moreover, we show that there exists a sequence of solutions to \eqref{sos} such that their strict one sub-level sets approximate the interior of the maximal robust domain of attraction in measure.

\begin{assumption}
\label{ass5}
One of the polynomials defining the set $D$ is equal to $h_i^D:=\|\bm{d}\|^2-R_D$ for some constant $R_D\geq 0$.
\end{assumption}
As argued in \cite{korda2014}, Assumption \ref{ass5} is without loss of generality since the set $D$ is compact, the redundant constraint $\|\bm{d}\|^2-R_D-1\leq 0$ can always be incorporated into the description of $D$ for sufficiently large $R_D$.

We in the following show that given an arbitrary $\epsilon>0$, there exists a polynomial solution $p(\bm{x})$ to \eqref{sos} such that $|p(\bm{x})-v(\bm{x})|<\epsilon$ holds for $\bm{x}\in B(\bm{0},R)$. Before this, we introduce a lemma from \cite{lin1996}.

\begin{lemma}[Lemma B.5 in \cite{lin1996}]
\label{smooth}
Let $B(\bm{0},R)$ be a compact subset in $\mathbb{R}^n$ and $u(\bm{x}): B(\bm{0},R)\rightarrow \mathbb{R}$ be a locally Lipschitz function. If there exists a continuous function $\rho: B(\bm{0},R) \rightarrow \mathbb{R}$ such that for each $\bm{d}\in D$,
\[\mathcal{L} u(\bm{x})\leq \rho(\bm{x}), \text{a.e.} \ \bm{x} \in B(\bm{0},R),\]
where $\mathcal{L} u(\bm{x})=\nabla_{\bm{x}} u(\bm{x}) \cdot \bm{f}(\bm{x},\bm{d})=\frac{\partial u}{\partial \bm{x}}\bm{f}(\bm{x},\bm{d})$
  (recall that $\nabla_{\bm{x}} u(\bm{x})$ is defined a.e., since $u$ is locally Lipschitz.),
 then for any given $\epsilon>0$, there exists some smooth function $\psi(\bm{x})$ defined on $B(\bm{0},R)$ such that
\[\sup_{\bm{x}\in B(\bm{0},R)}|\psi(\bm{x})-u(\bm{x})|<\epsilon\text{~and~}\sup_{\bm{d}\in  D}\mathcal{L} \psi(\bm{x})\leq \rho(\bm{x})+\epsilon\]
over $\bm{x}\in B(\bm{0},R)$.
\end{lemma}

\begin{theorem}
\label{existence}
Under Assumption \ref{ass5}, if $\delta$ is a positive value larger than or equal to one in \eqref{sos}, then for any $\epsilon>0$ there exists a polynomial solution $p(\bm{x})$ to \eqref{sos} such that $$0\leq p(\bm{x})-v(\bm{x})<\epsilon, \forall \bm{x}\in B(\bm{0},R).$$
\end{theorem}
\begin{proof}
When $\delta$ is a positive value larger than or equal to one, we have that $u(\bm{x})+(1-h_j(\bm{x}))-1 \geq u(\bm{x})+(1-h_j(\bm{x}))^{\delta}-1\geq 0$ for $\bm{x}\in \overline{\mathcal{X}\setminus \mathcal{X}_{\infty}}$, i.e., $u(\bm{x})-h_j(\bm{x})\geq u(\bm{x})+(1-h_j(\bm{x}))^{\delta}-1\geq 0$ for $\bm{x}\in \overline{\mathcal{X}\setminus \mathcal{X}_{\infty}}$, $j=1,\ldots, n_{\mathcal{X}}$. Also, since $v(\bm{x})$ in \eqref{v} satisfies \eqref{HJB2}, we have 
$v(\bm{x})$ satisfies \eqref{inequa2}. Consequently, for any $ \epsilon_1>0$, $v'(\bm{x})=v(\bm{x})+\epsilon_1$ satisfies the following constraints: 
\begin{equation*}
\label{upper111}
\left\{
\begin{array}{lll}
&-\frac{\partial v'(\bm{x})}{\partial \bm{x}}\bm{f}(\bm{x},\bm{d})-\delta q(\bm{x}) (1-v'(\bm{x}))\geq \delta \alpha \epsilon_1, \forall \bm{x}\in B(\bm{0},R)\setminus \mathcal{X}_{\infty}, \forall \bm{d}\in D,\\
&v'(\bm{x})-h_j(\bm{x})\geq \epsilon_1, j=1,\ldots,n_{\mathcal{X}},\forall \bm{x}\in \overline{\mathcal{X}\setminus \mathcal{X}_{\infty}}, \\
&v'(\bm{x})- 1\geq \epsilon_1, \forall \bm{x}\in B(\bm{0},R)\setminus\mathcal{X}.
\end{array}
\right.
\end{equation*}

According to Lemma \ref{conti}, $v(\bm{x})$ is locally Lipschitz continuous over $\bm{x}\in B(\bm{0},R)$. Therefore, $v'(\bm{x})$ is locally Lipschitz continuous over $\bm{x}\in B(\bm{0},R)$ as well. According to Lemma \ref{smooth}, we have that for any $\epsilon_2<\frac{\epsilon_1 \alpha \delta}{2}$ with $\epsilon_2>0$, there exists a continuous function $p'(\bm{x})$ such that
\[|p'(\bm{x})-v'(\bm{x})|< \epsilon_2, \forall \bm{x}\in B(\bm{0},R)\setminus \mathcal{X}_{\infty}\]
and
\[\sup_{\bm{d}\in D}\frac{\partial p'(\bm{x})}{\partial \bm{x}}\bm{f}(\bm{x},\bm{d})\leq -\delta q(\bm{x}) (1-v'(\bm{x}))-\delta \alpha \epsilon_1+\epsilon_2, \forall \bm{x}\in B(\bm{0},R)\setminus \mathcal{X}_{\infty}.\]

Since $B(\bm{0},R)$ is compact, there exists a polynomial $p(\bm{x})$ of sufficiently high degree such that 
\[\sup_{B(\bm{0},R)}|p(\bm{x})-2\epsilon_2-p'(\bm{x})|<\epsilon_2 \text{~and}\]
\[\sup_{B(\bm{0},R)\times D}|\frac{\partial p'(\bm{x})}{\partial \bm{x}}\bm{f}(\bm{x},\bm{d})-\frac{\partial p(\bm{x})}{\partial \bm{x}}\bm{f}(\bm{x},\bm{d})|<\epsilon_2.
\]
Then we have
$$0<p(\bm{x})-v(\bm{x})<\epsilon_1+4\epsilon_2,\forall \bm{x}\in B(\bm{0},R)\setminus \mathcal{X}_{\infty}$$ and $$\frac{\partial p(\bm{x})}{\partial \bm{x}}\bm{f}(\bm{x},\bm{d})<-\delta q(\bm{x}) (1-p(\bm{x})), \forall \bm{x}\in B(\bm{0},R)\setminus \mathcal{X}_{\infty},\forall \bm{d}\in D.$$ Thus, we have
\begin{equation*}
\label{upper112}
\left\{
\begin{array}{lll}
&-\frac{\partial p(\bm{x})}{\partial \bm{x}}\bm{f}(\bm{x},\bm{d})-\delta q(\bm{x}) (1-p(\bm{x}))>0, \forall \bm{x}\in B(\bm{0},R)\setminus \mathcal{X}_{\infty}, \forall \bm{d}\in D,\\
&p(\bm{x})-h_j(\bm{x})>0,j=1,\ldots,n_{\mathcal{X}},\forall \bm{x}\in \overline{\mathcal{X}\setminus \mathcal{X}_{\infty}},\\
&p(\bm{x})- 1>0, \forall \bm{x}\in B(\bm{0},R)\setminus\mathcal{X}.
\end{array}
\right.
\end{equation*}

The polynomial $p(\bm{x})$ is therefore strictly feasible in \eqref{sos}, which follows from the classical Putinar's Positivstellensatz \cite{putinar93}. Since $\epsilon_1$ is arbitrary and $\epsilon_2<\epsilon_1\alpha \delta$, the conclusion in Theorem \ref{existence} holds.
\end{proof}

Given $(\epsilon_k)_{k=1}^{\infty}$ with $\epsilon_k>0$ and $\lim_{k\rightarrow \infty}\epsilon_k=0$ with $k\in \mathbb{N}$, according to Theorem \ref{existence}, there exists a sequence
\begin{equation}
\label{p}
\big(p_k(\bm{x})\big)_{k=1}^{\infty}
\end{equation}
 satisfying \eqref{sos} such that $0\leq p_k(\bm{x})-v(\bm{x})< \epsilon_k.$
Denote
\begin{equation}
\label{rk}
\mathcal{R}_{k,0}:=\{\bm{x}\in B(\bm{0},R)\mid p_k(\bm{x})< 1\},
\end{equation}
 we next show that $\mathcal{R}_{k,0}$ inner-approximates the interior of the maximal robust domain of attraction in measure with $k$ approaching infinity.

\begin{theorem}
\label{convergence}
Let $\big(p_{k}(\bm{x})\big)_{k=1}^{\infty}$ and $\mathcal{R}_{k,0}$ be the sequence in \eqref{p} and the set in \eqref{rk} respectively. Then the set $\mathcal{R}_{k,0}$ converges to the interior of the maximal robust domain of attraction from inside in measure with $k$ tending towards infinity, i.e.,
\[ \lim_{k\rightarrow \infty} \mu(\mathcal{R}^{\circ}\setminus\mathcal{R}_{k,0})=0.\]
 \end{theorem}
\begin{proof}
Following the proof of Theorem 3 in \cite{lasserre2015}, we have that $\lim_{k\rightarrow \infty} \mu(\mathcal{R}_0\setminus\mathcal{R}_{k,0})=0$, where $\mathcal{R}_0=\{\bm{x}\in \mathbb{R}^n \mid v(\bm{x})<0\}$. According to Lemma \ref{relation}, we have that $\lim_{k\rightarrow \infty} \mu(\mathcal{R}^{\circ}\setminus\mathcal{R}_{k,0})=0.$
\end{proof}

\section{Examples and Discussions}
\label{ex}
In this section we illustrate our approach with five examples. All computations were performed on an i7-P51s 2.6GHz CPU with 4GB RAM running Windows 10. For the numerical implementation, we formulate the sum-of-squares problem \eqref{sos} using the Matlab package YALMIP \cite{lofberg2004} and employ Mosek of the academic version \cite{mosek2015mosek} as a semi-definite programming solver. The parameters that control the performance of our method are given in Table \ref{table}.

\oomit{In what follows, we will use $u_k\in \mathbb{R}_k[\bm{x}]$ to denote the concrete polynomial with degree $k$ used in \eqref{sos}, and 
$\mathcal{R}_{k,0}$ to denote the corresponding domain of attraction.}

\begin{table}[h!]
\begin{center}
\begin{tabular}{|l|r|r|r|r|r|r|r|}
  \hline
   Ex.&$k$&$\delta$&$\alpha$&$R$&$d_{s}$&$d_{s'}$&$T$\\\hline
   4.1&8&1&$10^{-4}$&$1.01$&10&8&1.18\\\hline
   4.1&10&1&$10^{-4}$&1.01&12&10 &1.20\\\hline
   4.1&16&1&$10^{-4}$&1.01&18&16 &1.39\\\hline
   4.1&24&1&$10^{-4}$&1.01&26&24&1.50\\\hline
   4.2&8&1&$10^{-2}$&1.211&10&8 & 0.68\\\hline
   4.2&16&1&$10^{-2}$&1.211&18&16&  8.80\\\hline
   4.3&4&1&$10^{-2}$&1.01&6&4& 1.40\\\hline
   4.3&8&1&$10^{-2}$&1.01&10&8& 0.68\\\hline
   4.3&12&1&$10^{-2}$&1.01&14&12 & 1.46\\\hline
   4.4&4&1&$10^{-2}$&1.01&6&4& 0.68 \\\hline
   4.4&6&1&$10^{-2}$&1.01&8&6 &0.82 \\\hline
   4.4&10&1&$10^{-2}$&1.01&12&10 &3.12\\\hline
   4.5&3&1&$10^{-2}$&1.01&4&2& 20.15 \\\hline
   4.5&4&1&$10^{-2}$&1.01&4&4 &35.56 \\\hline
   4.5&5&1&$10^{-2}$&1.01&6&4 &1257.20\\\hline
   \end{tabular}
\end{center}
\caption{\textit{Parameters that control the performance of the semi-definite program \eqref{sos} on the examples presented in this section. $\delta, \alpha$ and $R$ are the constant values in \eqref{sos}, where $\mathcal{X}_{\infty}=\{\bm{x}\mid q(\bm{x})<\alpha\}$ with $q(\bm{x})=\|\bm{x}\|^2$ and $B(\bm{0},R)=\{\bm{x}\mid R-\|\bm{x}\|^2\geq 0\}$. $k$, $d_{s}$ and $d_{s'}$ denote the degree of the polynomials $u_k, \{s_{0},s_1,s_{2,i},i=1,\ldots,m_D,s_3\}$ and $\{s_{4,j},s_{5,j},s_{6,j},s_{7,j},s_{8,j},s_{9,j},s_{10,l,j},j,l=1,\ldots,n_{\mathcal{X}}\}$ in \eqref{sos}, respectively; $T$: computation times (seconds).} }
\label{table}
\end{table}

\begin{example}
\label{ucd}
Consider a one-dimensional system adapted from \cite{henrion2014}, which is given by
\[
\dot{x}=x(x-(d+0.5))(x+0.5)
\]
with $\mathcal{X}=\{x\in \mathbb{R}\mid x^2<1\}$ and $D=\{d\in \mathbb{R}\mid 100d^2-1\leq 0\}$.

The origin is a locally uniformly exponentially stable state. The maximal robust domain of attraction in this case is determined analytically as $\mathcal{R}=(-0.5,0.4)$. \oomit{Actually, $q(\bm{x})=x^2$ is a local Lyapunov function for this system.} Theorem \ref{inner} indicates that the strict one sub-level set  of the approximating polynomial $u$ computed by solving \eqref{sos} is a robust domain of attraction. Plots of the computed robust domains of attraction for approximating polynomials of degree $k=8,10,16,24$ are presented in Fig. \ref{fig-one-00}. The visualized results in Fig. \ref{fig-one-00} further confirm that the strict one sub-level set of the approximating polynomial $u(\bm{x})$ returned by solving \eqref{sos} is indeed a robust domain of attraction. The relative volume errors, which are computed approximately by Monte Carlo integration, are also reported in Table \ref{table0}. From Fig. \ref{fig-one-00} and Table \ref{table0}, we observe fairly good tightness of the estimates since $k=10$. 

\begin{table}[h!]
\begin{center}
\begin{tabular}{|l|r|r|r|r|}
  \hline
      $k$&$8$&$10$&$16$&$24$\\\hline
       error &13.2\%& 4.48\%& 3.41\%& 2.85\%\\\hline
   \end{tabular}
\end{center}
\caption{Relative error estimations of computed robust domains of attraction to the maximal robust domain of attraction as a function of the approximating polynomial degree for Example \ref{ucd}. }
\label{table0}
\end{table}

\begin{figure*}[h]
\centering
\setlength\fboxsep{0pt}
\setlength\fboxrule{0.15pt}
\begin{tabular}{cccc}
\fbox{\includegraphics[width=2.2in,height=1.10in]{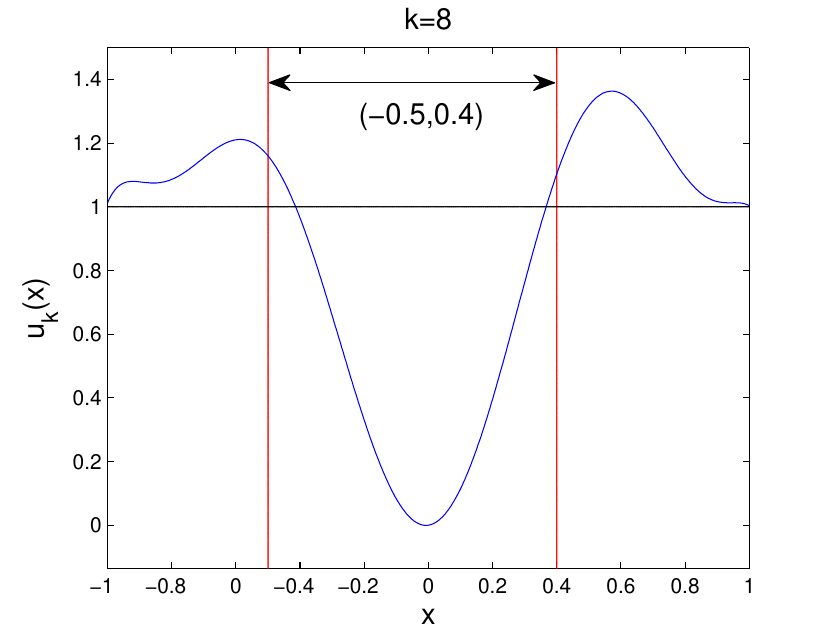}}
\fbox{\includegraphics[width=2.2in,height=1.1in]{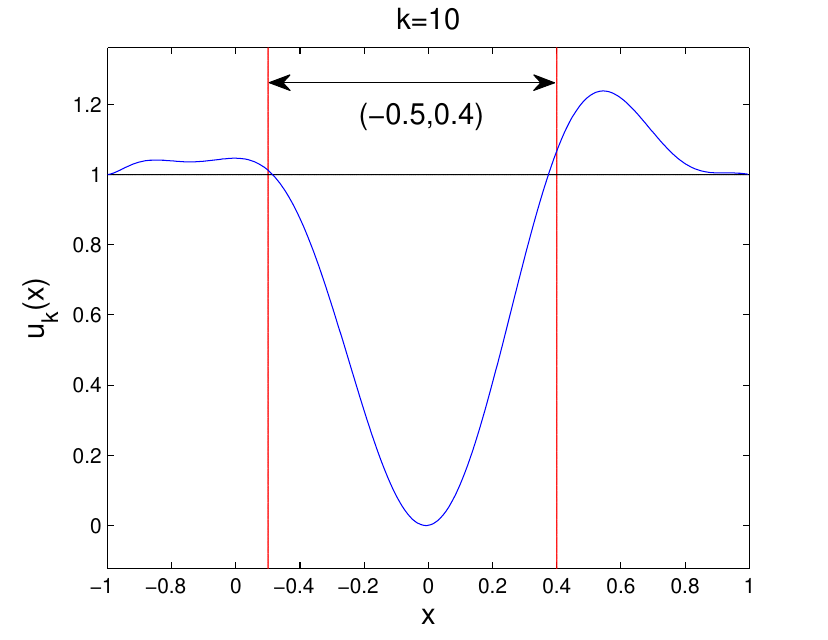}}
\\
\fbox{\includegraphics[width=2.2in,height=1.1in]{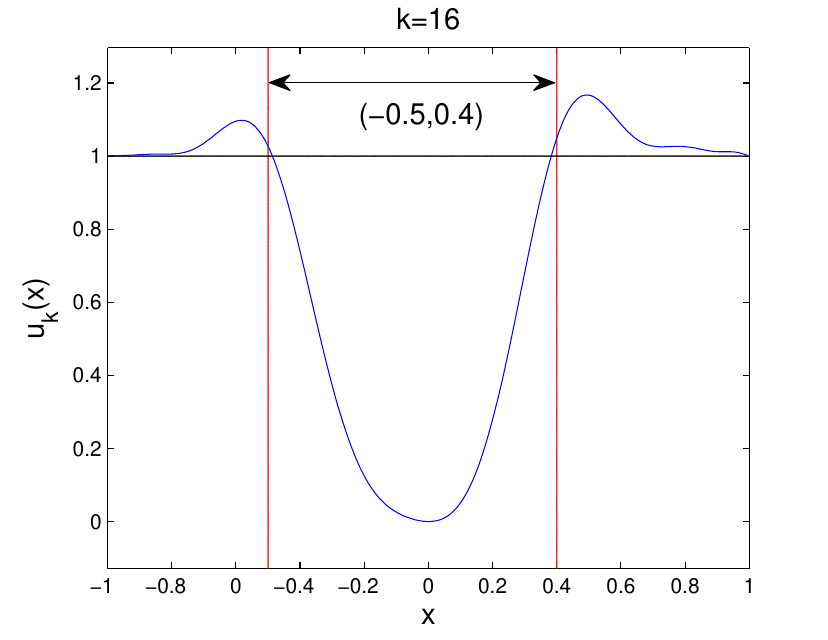}}
\fbox{\includegraphics[width=2.2in,height=1.1in]{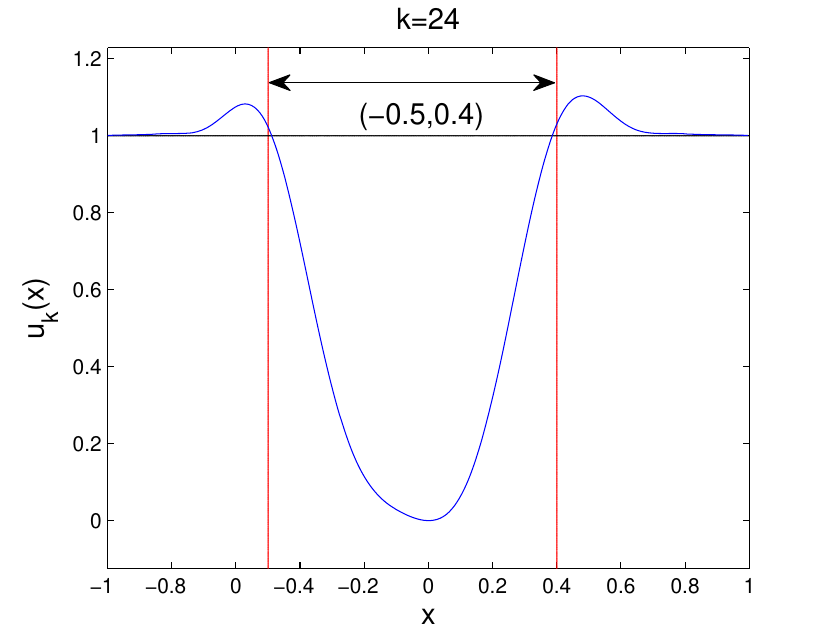}}
\end{tabular}
\caption{An illustration of computed robust domains of attraction for Example \ref{ucd}. Blue curve denotes the level sets of the approximating polynomial $u(\bm{x})$. Red curve denotes the boundary of the maximal robust domain of attraction.}
\label{fig-one-00}
\end{figure*}

\end{example}
\begin{example}
\label{ex0}
The second example considers scaled version of the reversed-time Van der Pol oscillator free of perturbations \cite{henrion2014} given by 
\begin{equation*}
\label{ex10}
\begin{array}{lll}
\dot{x}=-2y,\\
\dot{y}=0.8x+10(x^2-0.21)y,
\end{array}
\end{equation*}
with $\mathcal{X}=\{\bm{x}\mid \frac{x^2+y^2}{1.21}<1\}$.

The origin is a locally uniformly exponentially stable state. For this example, there exists a limit cycle, which is the boundary of the maximal robust domain of attraction. The limit cycle is included in $\mathcal{X}$. \oomit{As in Example \ref{ucd}, $q(\bm{x})$ is a local Lyapunov function for this system.}Theorem \ref{inner} indicates that the strict one sub-level set of the approximating polynomial $u(\bm{x})$ computed by solving \eqref{sos} is a robust domain of attraction. Plots of computed robust domains of attraction for approximating polynomials of degree $k =8, 16$, are shown in Fig. \ref{fig-one-0}. The visualized results in Fig. \ref{fig-one-0} further confirm that the strict one sub-level set of the approximating polynomial $u(\bm{x})$ returned by solving \eqref{sos} is indeed a robust domain of attraction. In order to quantitatively assess the quality of computed robust domains of attraction, we use the simulation technique to synthesize an estimate $\tilde{\mathcal{R}}$ of the maximal robust domain of attraction by gridding the state space, and then compute the relative volume errors approximately according to the formula
$(1-\frac{\text{number of grid states in }\mathcal{R}_{k}}{\text{number of grid states in }\tilde{\mathcal{R}}})\times 100\%$, where $\mathcal{R}_k$ is the robust domain of attraction formed by the approximating polynomial of degree $k$. The estimate $\tilde{\mathcal{R}}$ is shown in Fig. \ref{fig-one-0}. The visualized results in Fig. \ref{fig-one-0} indicate that the estimate $\tilde{\mathcal{R}}$ approximates the maximal robust domain of attraction quite well. The relative volume errors are reported in Table \ref{table1}. From Fig. \ref{fig-one-0} and Table \ref{table1} we observe fairly good tightness of estimates since $k=8$.

\begin{table}[h!]
\begin{center}
\begin{tabular}{|l|r|r|r|r|r|r|}
  \hline
      $k$&$8$& $16$\\\hline
       error& 7.99\%& 5.84\%\\\hline
   \end{tabular}
\end{center}
\caption{Relative volume error estimations of computed robust domains of attraction to the maximal robust domain of attraction as a function of the approximating polynomial degree for Example \ref{ex0}. }
\label{table1}
\end{table}

\begin{figure*}[h]
\centering
\setlength\fboxsep{0pt}
\setlength\fboxrule{0.15pt}
\begin{tabular}{cccc}
\fbox{\includegraphics[width=1.6in,height=1.5in]{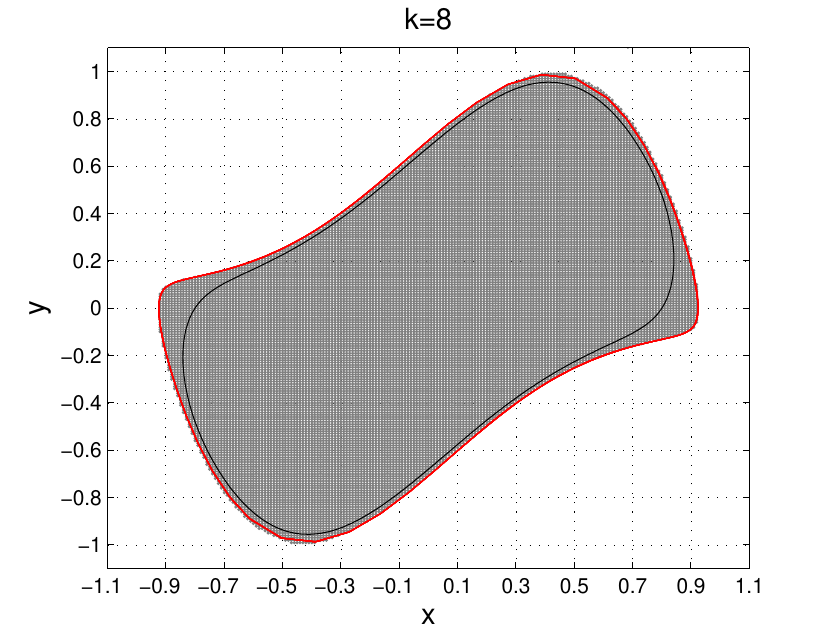}}
\fbox{\includegraphics[width=1.6in,height=1.5in]{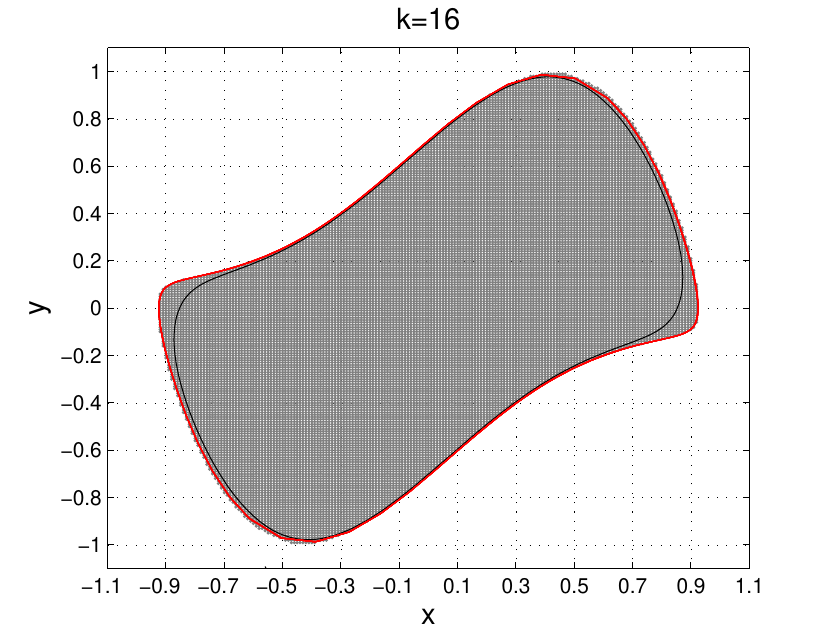}}
\end{tabular}
\caption{An illustration of computed robust domains of attraction for Example \ref{ex0}. Black curve denotes the boundary of the computed robust domain of attraction. Red curve denotes the limit cycle, which is the boundary of the maximal robust domain of attraction. Gray region denotes an estimate $\tilde{\mathcal{R}}$ of the maximal robust domain of attraction, which is computed using simulation techniques.}
\label{fig-one-0}
\end{figure*}
\end{example}

\begin{example}
\label{chemical}
In this example we consider a chemical oscillator from \cite{papachristodoulou2002}. The simplest, but chemically plausible trimolecular reaction is 
\[X \overset{k_1}{\underset{k_{-1}}{\rightleftharpoons}} A, B\overset{k_2}{\rightarrow} Y, 2X+Y\overset{k_3}{\rightarrow} 3X,\]
in which species $X$ is in dynamical equilibrium with species $A$ with a forward rate of reaction $k_1$ and a backward rate of reaction $k_{-1}$, and so on. Using the law of mass action, and non-dimensionalising the equations, we have
\begin{equation*}
\label{ex100}
\begin{array}{lll}
\dot{x}=a-x+x^2y,\\
\dot{y}=b-x^2y,
\end{array}
\end{equation*}
where $x$, $y$ are the non-dimensional concentrations of $X$ and $Y$, and $a$, $b$ are non-negative constant parameters that depend on the concentrations of $A$ and $B$. 

Like \cite{papachristodoulou2002}, we take $a=0.5$ and $b=0.5$. The system has a locally uniformly exponentially stable state $(1,0.5)$. Since $x$ and $y$ are the non-dimensional concentrations of $X$ and $Y$, they are naturally positive, i.e. $x> 0$ and $y> 0$. On the other hand, they should have upper bound on the concentrations. In this paper we impose the inequality constraint $(x-1)^2+4(y-0.5)^2<1$.

After translating the equilibrium $(1,0.5)$ to the origin $(0,0)$ and making $x_1=x,x_2=2y$, we obtain the equivalent system of interest in this example, 
\begin{equation}
\label{ex1000}
\begin{split}
&\dot{x}_1=0.5-(x_1+1)+(x_1+1)^2(\frac{x_2}{2}+0.5)\\
&\dot{x}_2=1-(x_1+1)^2(x_2+1)
\end{split}
\end{equation}
with $\mathcal{X}=\{\bm{x}\mid x_1^2+x_2^2<1\}$.

The origin is a locally uniformly exponentially stable state for system \eqref{ex1000}. Theorem \ref{inner} indicates that the strict one sub-level set of the approximating polynomial $u(\bm{x})$ computed by solving \eqref{sos} is a robust domain of attraction. Plots of computed robust domains of attraction for approximating polynomials of degree $k =4, 8, 12$, are shown in Fig. \ref{fig-one-0_che}. The visualized results in Fig. \ref{fig-one-0_che} further confirm that the strict one sub-level set of the approximating polynomial $u(\bm{x})$ returned by solving \eqref{sos} is indeed a robust domain of attraction. Like Example \ref{ex0}, we use the simulation technique to quantitatively assess the quality of computed robust domains of attraction. The relative volume errors are listed in Table \ref{table1_che}. From Fig. \ref{fig-one-0_che} and Table \ref{table1_che} we observe fairly good tightness of estimates since $k=8$. 

\begin{figure*}[h]
\centering
\setlength\fboxsep{0pt}
\setlength\fboxrule{0.15pt}
\begin{tabular}{cccc}
\fbox{\includegraphics[width=1.5in,height=1.2in]{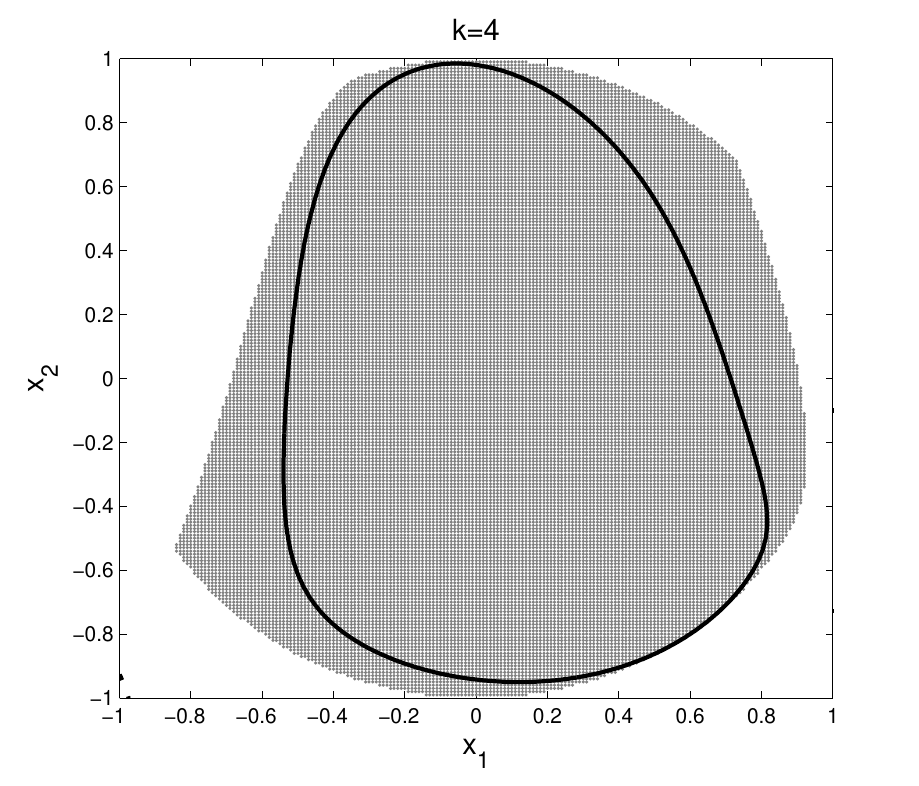}}
\fbox{\includegraphics[width=1.5in,height=1.2in]{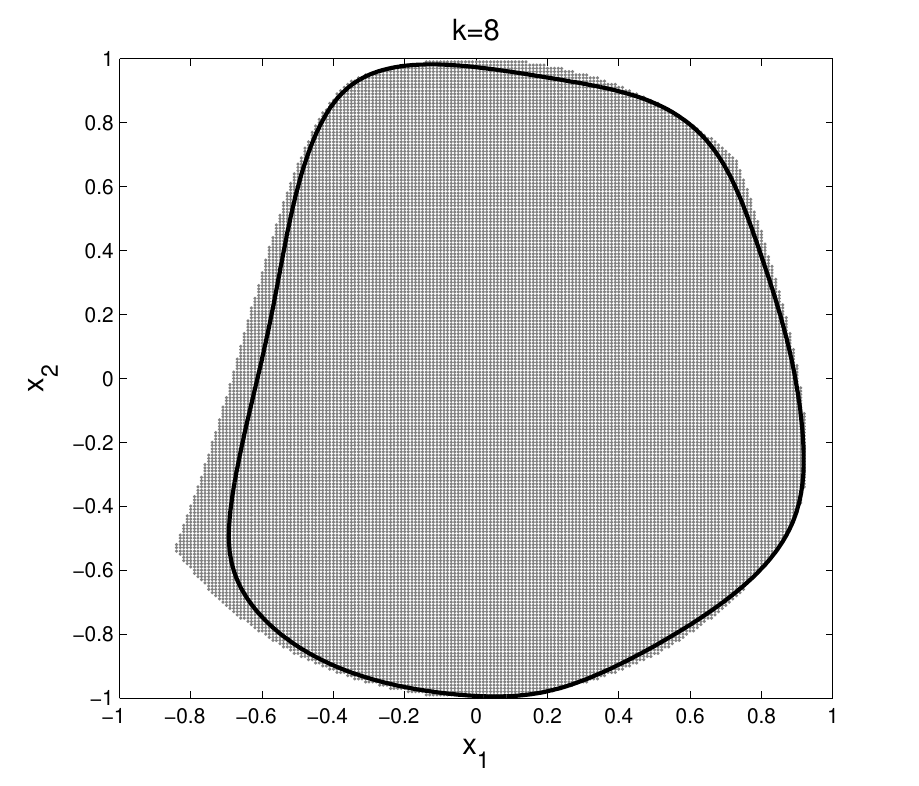}}
\fbox{\includegraphics[width=1.5in,height=1.2in]{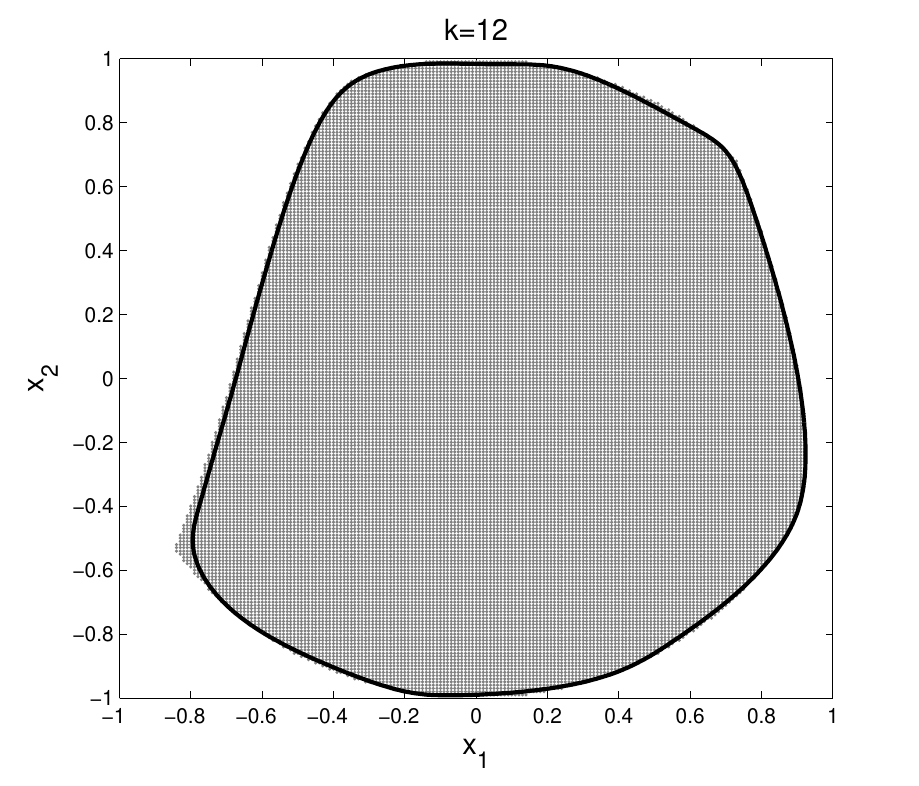}}
\end{tabular}
\caption{An illustration of computed robust domains of attraction for Example \ref{chemical}. Black curve denotes the boundary of the computed robust domain of attraction. Gray region denotes an estimate of the maximal robust domain of attraction, which is computed using simulation techniques.}
\label{fig-one-0_che}
\end{figure*}

\begin{table}[h!]
\begin{center}
\begin{tabular}{|l|r|r|r|r|}
  \hline
      $k$&$4$&$8$&$12$\\\hline
       error &25.02\%& 6.47\%& 1.75\%\\\hline
   \end{tabular}
\end{center}
\caption{Relative volume error estimations of computed robust domains of attraction to the maximal robust domain of attraction as a function of the approximating polynomial degree for Example \ref{chemical}. }
\label{table1_che}
\end{table}

\end{example}

\begin{example}
\label{ex1}
Consider a system from \cite{grune2015}, whose dynamics are described by
\begin{equation}
\label{ex11}
\begin{array}{lll}
\dot{x}=-x+y,\\
\dot{y}=-\frac{1}{10}x-2y-x^2+(d+\frac{1}{10})x^3.
\end{array}
\end{equation}

The origin for system \eqref{ex11} is locally uniformly exponentially stable. In this example $D=[4.9,5.1]$ and $\mathcal{X}=\{\bm{x}\mid x^2+y^2<1\}$. In order to fit \eqref{sos}, we reformulate \eqref{ex11} as the following equivalent system 
\begin{equation*}
\label{ex12}
\begin{array}{lll}
\dot{x}=-x+y,\\
\dot{y}=-\frac{1}{10}x-2y-x^2+(d+5+\frac{1}{10})x^3,
\end{array}
\end{equation*}
where $D=\{d\in \mathbb{R}\mid 100d^2-1\leq 0\}$ and $\mathcal{X}=\{\bm{x}\mid x^2+y^2<1\}$.

Theorem \ref{inner} indicates that the strict one sub-level set of the solution to \eqref{sos} is a robust domain of attraction. Plots of computed robust domains of attraction $\mathcal{R}_{k}$, $k=4,6,10$, are shown in Fig. \ref{fig-one-1}. We also give an estimate of the maximal robust domain of attraction by simulation methods and estimate the relative volume errors as in Example \ref{ex0}. From Table \ref{table2}, which lists the relative volume errors, we observe fairly good tightness of the estimates since $k=4$. 

\begin{table}[h!]
\begin{center}
\begin{tabular}{|l|r|r|r|r|r|}
  \hline
      $k$&$4$&$6$& $10$\\\hline
       error & 8.88\%&6.94\%& 4.98\%\\\hline
   \end{tabular}
\end{center}
\caption{Relative volume error estimations of computed robust domains of attraction to the maximal robust domain of attraction as a function of the approximating polynomial degree for Example \ref{ex1}. }
\label{table2}
\end{table}

\begin{figure*}[h]
\centering
\setlength\fboxsep{0pt}
\setlength\fboxrule{0.15pt}
\begin{tabular}{cccc}
\fbox{\includegraphics[width=1.6in,height=1.2in]{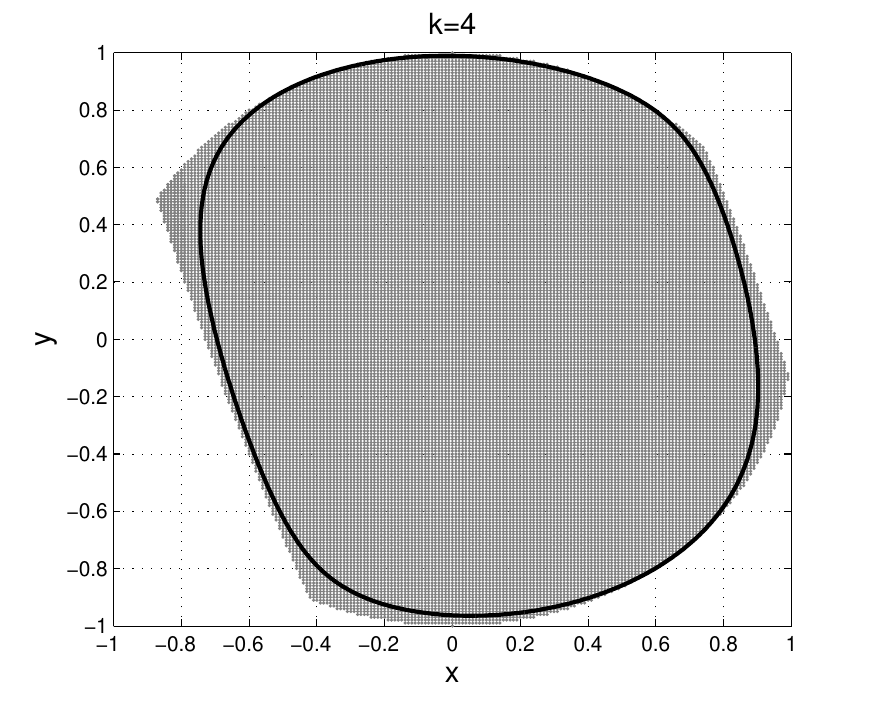}}
\fbox{\includegraphics[width=1.6in,height=1.2in]{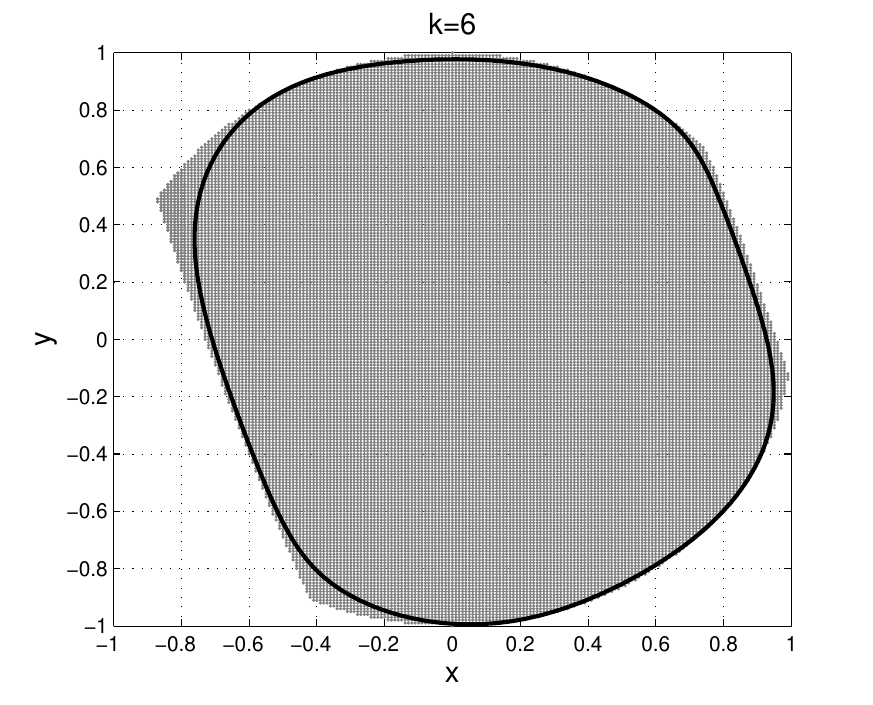}}
\fbox{\includegraphics[width=1.6in,height=1.2in]{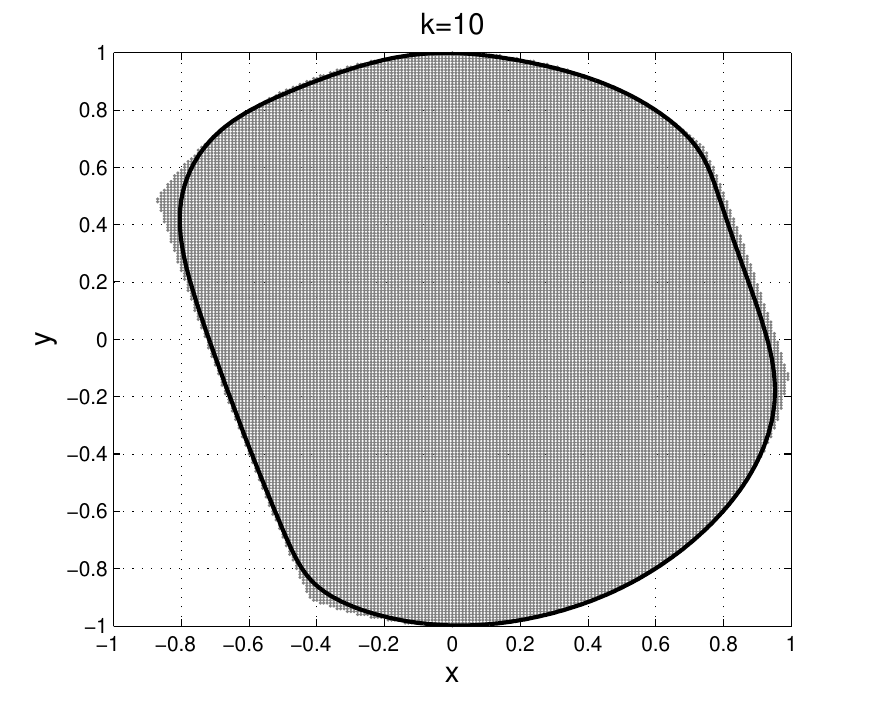}}
\end{tabular}
\caption{An illustration of computed robust domains of attraction for Example \ref{ex1}. Black curve denotes the boundary of the computed domain of attraction. Gray region denotes an estimate of the maximal robust domain of attraction, which is computed using simulation techniques.}
\label{fig-one-1}
\end{figure*}
\end{example}

\begin{example}
\label{ex2}
Consider a seven-dimensional system, which is mainly employed to illustrate the scalability issue of our semi-definite programming based method in dealing with high dimensional system.
\begin{equation}
\label{ex21}
\begin{array}{lll}
&\dot{x}_1=-x_1+0.5x_2,\\
&\dot{x}_2=-x_2+0.4x_3,\\
&\dot{x}_3= -x_3+0.5x_4,\\
&\dot{x}_4= -x_4+0.7x_5,\\
&\dot{x}_5=-x_5+0.5x_6,\\
&\dot{x}_6= -x_6+0.8x_7,\\
&\dot{x}_7= -x_7+10x_1^2+dx_2^2-x_3^2-x_4^2+x_5^2,\\
\end{array}
\end{equation}
where $D=\{d\in \mathbb{R}\mid d^2+0.75-1\leq 0\}$ and $\mathcal{X}=\{\bm{x}\mid \|\bm{x}\|^2<1\}$.

The equilibrium state $\bm{0}$ is locally uniformly exponentially stable. Theorem \ref{inner} indicates that the strict one sub-level set of the solution to \eqref{sos} is a robust domain of attraction. Plots of computed robust domains of attraction for approximating polynomials of degree $k =3,4,5$, on planes $x_1-x_2$ with $x_3=x_4=x_5=x_6=x_7=0$ and $x_1-x_7$ with $x_2=x_3=x_4=x_5=x_6=0$ are shown in Fig. \ref{fig-one-2}. 
In order to shed light on the accuracy of the computed domains of attraction, we use the simulation technique to synthesize coarse estimations of the maximal robust domain of attraction on planes $x_1-x_2$ with $x_3=x_4=x_5=x_6=x_7=0$ and $x_1-x_7$ with $x_2=x_3=x_4=x_5=x_6=0$ by taking initial states in the state spaces $\{\bm{x}\mid \|\bm{x}\|^2<1\wedge x_3=0\wedge x_4=0\wedge x_5=0\wedge x_6=0\wedge x_7=0\}$ and $\{\bm{x}\mid \|\bm{x}\|^2<1\wedge x_2=0\wedge x_3=0\wedge x_4=0\wedge x_5=0\wedge x_6=0\}$, respectively. They are the gray regions in Fig. \ref{fig-one-2}. The relative volume errors on both planes $x_1-x_2$ with $x_3=x_4=x_5=x_6=x_7=0$ and $x_1-x_7$ with $x_2=x_3=x_4=x_5=x_6=0$, which are computed following the way in Example \ref{ex0}, are reported in Table \ref{table71}. 

\begin{table}[h!]
\begin{center}
\begin{tabular}{|l|r|r|r|}
  \hline
      $k$&$3$&$4$&$5$\\\hline
       error($x_1-x_2$) &32.17\%& 16.83\%& 10.57\%\\\hline
       error($x_1-x_7$) &32.35\%& 20.40\%& 11.21\%\\\hline
   \end{tabular}
\end{center}
\caption{Relative volume error estimations of computed robust domains of attraction to the maximal robust domain of attraction on planes $x_1-x_2$ with $x_3=x_4=x_5=x_6=x_7=0$ and $x_1-x_7$ with $x_2=x_3=x_4=x_5=x_6=0$ as a function of the approximating polynomial degree for Example \ref{ex2}. }
\label{table71}
\end{table}

\begin{figure}
\centering
\setlength\fboxsep{0pt}
\setlength\fboxrule{0.15pt}
\begin{tabular}{cccc}
\fbox{\includegraphics[width=2.5in,height=1.5in]{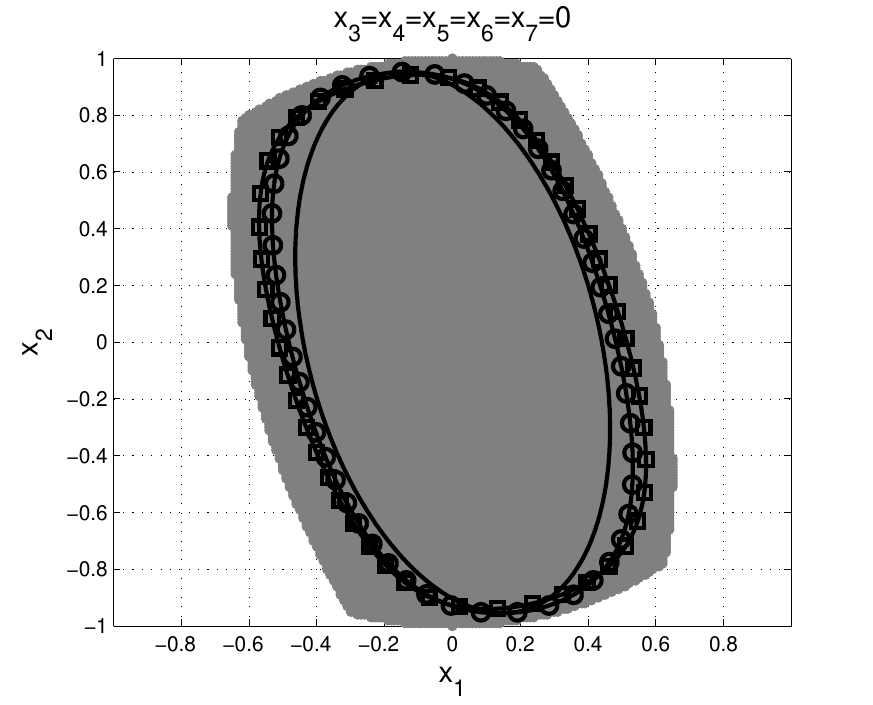}}
\fbox{\includegraphics[width=2.5in,height=1.5in]{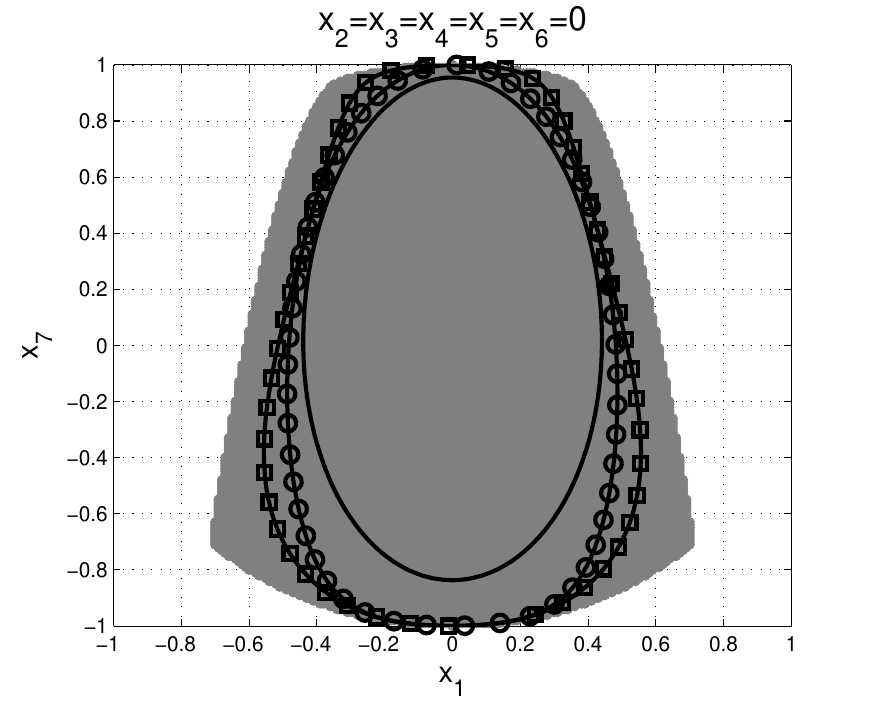}}
\end{tabular}
\caption{An illustration of computed robust domains of attraction for Example \ref{ex2}. Black curve with square marker, black curve with circle marker and black curve denote the boundaries of the robust domains of attraction computed when $k=5$, $4$ and $3$ respectively. Gray region denotes an estimate of the maximal robust domain of attraction.}
\label{fig-one-2}
\end{figure}
\end{example}

Based on Examples \ref{ucd}$\sim$\ref{ex2}, we conclude that approximating polynomials of higher degree would return less conservative robust domains of attraction. Although the size of the semi-definite program in \eqref{sos} grows extremely
fast with the number of state and perturbation variables and the degree of the polynomials in \eqref{sos}, it is worth emphasizing that we are dealing with nonlinear non-convex infinite-dimensional problems by solving a semi-definite programming problem, which is relatively simple to implement. Yet, despite the difficulty of the problems considered, the constructed semi-definite program \eqref{sos} possess solutions whose strict one sub-level sets inner-approximate the interior of the maximal robust domain of attraction in measure under appropriate assumptions according to Theorem \ref{convergence}. In order to improve the scalability issue of our method and further apply it to higher dimensional systems, some techniques such as exploiting the algebraic structure \cite{parrilo2005} of the semi-definition programming \eqref{sos} and using template polynomials such as (scaled-) diagonally-dominant-sums-of-squares polynomials \cite{majumdar2014,Ahmadi17,ahmadi2017} would facilitate such gains. 

In the rest we further give a brief discussion on the other parameters that control the performance of the semi-definite program \eqref{sos} in terms of relative volume error estimations of computed robust domains of attraction to the maximal robust domain of attraction. The parameters related to the degree of polynomials in Table \ref{table} are already discussed either above or in the existing literature pertinent to the sum-of-squares optimization:  polynomials of higher degree in \eqref{sos} would return less conservative robust domains of attraction generally. Thus, we just discuss the parameters $\alpha$, $\delta$ and $R$. Finding an optimal combination of these parameters to obtain the least conservative estimates is not the focus of this discussion and would be discussed in detail in the future work. The parameters $\alpha$ and $R$ could reflect the size of the sets $\mathcal{X}_{\infty}$ and $B(\bm{0},R)$ in \eqref{sos}, respectively. The analysis on these three parameters is based on the other parameter values, i.e., the degrees of polynomials, listed in Table \ref{table}, and is summarized in Table \ref{table01} $\sim$ \ref{table03}. ``-" in Table \ref{table01} $\sim$ \ref{table03} means that \eqref{sos} did not return a feasible solution and thus we did not obtain the relative volume error estimation. For the convenience of comparison we also add the relative volume error estimations listed in Table \ref{table0} $\sim$ \ref{table71} into Table \ref{table01} $\sim$ \ref{table03}.

Table \ref{table01} $\sim$ \ref{table03} indicate that the parameters $\delta$, $\alpha$ and $R$ definitely influence the performance of \eqref{sos} for some cases. Increasing $\delta$ and $\alpha$ and/or decreasing $R$ may result in less conservative estimates, but this does not hold always. This influence is weakening with the degree of approximating polynomials increasing.

\begin{table*}[t!]
\begin{center}
\begin{tabular}{|p{1cm}|p{1cm}|p{1cm}|p{1cm}|p{1cm}|p{1cm}|p{1cm}|}
  \hline
      \multirow{3}{*}{Ex. 4.1}& &$k$&$8$&$10$&$16$&$24$\\
                        &$\delta=3$&error &2.93\%& 2.11\%& 1.75\%& 1.73\%\\
                        &$\delta=1$&error&$13.2\%$&4.48\%&3.41\%&2.85\%\\\hline
              
   \end{tabular}
\end{center}
\begin{center}
\begin{tabular}{|p{1cm}|p{1.68cm}|p{1.68cm}|p{1.75cm}|p{1.75cm}|}
  \hline
      \multirow{3}{*}{Ex. 4.2 }& &$k$&$8$&$16$\\
                        &$\delta=3$&error &4.86\%& 0.53\%\\
                        &$\delta=1$&error &7.99\%& 3.27\%\\\hline
   \end{tabular}
\end{center}
\begin{center}
\begin{tabular}{|p{1cm}|p{1.25cm}|p{1.25cm}|p{1.25cm}|p{1.4cm}|p{1.25cm}|}
  \hline
      \multirow{3}{*}{Ex. 4.3 }&&$k$&$4$&$8$&$12$\\
                       &$\delta=3$&error & -& 33.35\%& 13.13\%\\
                       
                        &$\delta=1$&error & 25.02\%& 6.47\%& 1.75\%\\\hline
   \end{tabular}
\end{center}
\begin{center}
\begin{tabular}{|p{1cm}|p{1.25cm}|p{1.25cm}|p{1.25cm}|p{1.4cm}|p{1.25cm}|}
  \hline
      \multirow{3}{*}{Ex. 4.4 }&&$k$&$4$&$6$&$10$\\
                       &$\delta=3$&error &9.04\%& 6.88\%& 5.35\%\\
                       
                       &$\delta=1$&error&  8.88\%&6.94\%& 4.98\%\\\hline
   \end{tabular}
\end{center}
\begin{center}
\begin{tabular}{|p{1.0cm}|p{0.95cm}|c|c|c|c|}
  \hline
      \multirow{5}{*}{Ex. 4.5}& &$k$&$3$&$4$&$5$\\
                       &\multirow{2}{*}{$\delta=3$} &error($x_1-x_2$) &55.99\%& 42.51\%& 38.16\%\\
       & &error($x_1-x_7$) &57.56\%&48.61\%& 39.47\%\\
        &\multirow{2}{*}{$\delta=1$} &error($x_1-x_2$)  &32.17\%& 16.83\%& 10.57\%\\
          & &error($x_1-x_7$) &32.35\%& 20.40\%&11.21\%\\\hline                                
   \end{tabular}
\end{center}
\caption{An illustration for relative volume error estimations affected by $\delta$. }
\label{table01}
\end{table*}

\begin{table*}[t!]
\begin{center}
\begin{tabular}{|p{1cm}|p{1.5cm}|p{1cm}|p{1cm}|p{1cm}|p{1cm}|p{1cm}|}
  \hline
      \multirow{3}{*}{Ex. 4.1}&&$k$&$8$&$10$&$16$&$24$\\
                        &$\alpha=0.01$ &error &5.73\%& 4.06\%& 3.35\%& 2.76\%\\
                        &$\alpha=10^{-4}$ &error &13.2\%& 4.48\%& 3.41\%& 2.85\%\\        
                        
                        \hline
   \end{tabular}
\end{center}
\begin{center}
\begin{tabular}{|p{1cm}|p{2.4cm}|p{1.65cm}|p{1.65cm}|p{1.65cm}|}
  \hline
      \multirow{3}{*}{Ex. 4.2}&&$k$&$8$&$16$\\
                        &$\alpha=0.1$&error &8.04\%& 3.39\%\\
                        &$\alpha=0.01$&error &7.99\%& 3.27\%\\\hline
   \end{tabular}
\end{center}
\begin{center}
\begin{tabular}{|p{1cm}|p{1.5cm}|p{1.35cm}|p{1.35cm}|p{1.35cm}|p{1.35cm}|}
  \hline
      \multirow{3}{*}{Ex. 4.3}&&$k$&$4$&$8$&$12$\\
                        &$\alpha=0.1$&error &21.00\%& 5.90\%& 1.78\%\\
                        &$\alpha=0.01$&error &25.02\%& 6.47\%& 1.75\%\\\hline
   \end{tabular}
\end{center}
\begin{center}
\begin{tabular}{|p{1cm}|p{1.5cm}|p{1.35cm}|p{1.35cm}|p{1.35cm}|p{1.35cm}|}
  \hline
      \multirow{3}{*}{Ex. 4.4 }&&$k$&$4$&$6$&$10$\\
                       &$\alpha=0.1$&error &8.13\%& 5.57\%& 3.49\%\\
                       
                             &$\alpha=0.01$& error & 8.88\%&6.94\%& 4.98\%\\\hline
   \end{tabular}
\end{center}
\begin{center}
\begin{tabular}{|p{1cm}|p{1.45cm}|c|c|c|c|}
  \hline
      \multirow{5}{*}{Ex. 4.5}& &$k$&$3$&$4$&$5$\\
                      
       & \multirow{2}{*}{$\alpha=0.1$}&error($x_1-x_2$) &32.63\%& 16.80\%&10.10\%\\
        & &error($x_1-x_7$) &32.36\%& 20.13\%& 10.86\%\\
         &\multirow{2}{*}{$\alpha=0.01$} &error($x_1-x_2$)  &32.17\%& 16.83\%& 10.57\%\\
        & &error($x_1-x_7$) &32.35\%& 20.40\%&11.21\%\\\hline                 
   \end{tabular}
\end{center}
\caption{An illustration for relative volume error estimations affected by $\alpha$.}
\label{table02}
\end{table*}

\begin{table*}[t!]
\begin{center}
\begin{tabular}{|p{1cm}|p{1.5cm}|p{1cm}|p{1cm}|p{1cm}|p{1cm}|p{1cm}|}
  \hline
      \multirow{3}{*}{Ex. 4.1}&&$k$&$8$&$10$&$16$&$24$\\
                         &$R=1.1$&error&14.55\%& 7.58\%& 3.75\%& 2.82\%\\
                         
                   &$R=1.01$ &error &13.2\%& 4.48\%& 3.41\%& 2.85\%\\\hline
   \end{tabular}
\end{center}
\begin{center}
\begin{tabular}{|p{1cm}|p{2.4cm}|p{1.65cm}|p{1.65cm}|p{1.65cm}|}
  \hline
      \multirow{3}{*}{Ex. 4.2}&&$k$&$8$&$16$\\
                     &$R=1.5$&error&12.29\%& 6.10\%\\
                     &$R=1.211$&error &7.99\%& 3.27\%\\\hline
   \end{tabular}
\end{center}
\begin{center}
\begin{tabular}{|p{1cm}|p{1.5cm}|p{1.35cm}|p{1.35cm}|p{1.35cm}|p{1.35cm}|}
  \hline
      \multirow{3}{*}{Ex. 4.3}& &$k$&$4$&$8$&$12$\\
                     &$R=1.5$&error&-& 7.28\%& 3.95\%\\
                     &$R=1.01$&error &25.02\%& 6.47\%& 1.75\%\\\hline
   \end{tabular}
\end{center}
\begin{center}
\begin{tabular}{|p{1cm}|p{1.5cm}|p{1.35cm}|p{1.35cm}|p{1.35cm}|p{1.35cm}|}
  \hline
      \multirow{3}{*}{Ex. 4.4 }&&$k$&$4$&$8$&$12$\\
             &$R=1.5$&error & 8.20\%& 5.06\%& 4.46\%\\
                       
                             &$R=1.01$& error & 8.88\%&6.94\%& 4.98\%\\\hline
   \end{tabular}
\end{center}
\begin{center}
\begin{tabular}{|p{1cm}|p{1.45cm}|c|c|c|c|}
  \hline
      \multirow{5}{*}{Ex. 4.5}&&$k$&$3$&$4$&$5$\\
                       &\multirow{2}{*}{$R=1.5$}&error($x_1-x_2$) &32.76\%& 18.21\%& 9.77\%\\
      & &error($x_1-x_7$) &32.79\%& 20.64\%& 9.99\%\\
      &\multirow{2}{*}{$R=1.01$} &error($x_1-x_2$)  &32.17\%& 16.83\%& 10.57\%\\
       & &error($x_1-x_7$) &32.35\%& 20.40\%&11.21\%\\\hline                 
   \end{tabular}
\end{center}
\caption{An illustration for relative volume error estimations affected by $R$. }
\label{table03}
\end{table*}

\section{Conclusion}
\label{con}
In this paper a semi-definite programming based method was proposed for synthesizing robust domains of attraction for state-constrained perturbed polynomial systems. The semi-definite program, which falls within the convex programming framework and can be solved in polynomial time via interior-point methods, was constructed from a generalized Zubov's equation. Under appropriate  assumptions the existence of solutions to the constructed semi-definite program is guaranteed and there exists a sequence of solutions such that their strict one sub-level sets inner-approximate the interior of the maximal robust domain of attraction in measure. Finally, we evaluated the performance of the method on five examples.

\bibliographystyle{siamplain}
\bibliography{reference}

\end{document}